\documentclass{article}

\usepackage{amsmath,amssymb,amsthm,bbm,graphicx,url,color}
\usepackage{vmargin}

\setpapersize{A4}

\newcommand{\g}[1]{\boldsymbol{#1}}

\newcommand{\R}[0]{\mathbb{R}} 
\newcommand{\C}[0]{\mathbb{C}} 

\renewcommand{\i}[0]{\mathrm{i}} 

\newtheorem{theorem}{Theorem}
\newtheorem{lemma}{Lemma}

\newtheorem{definition}{Definition}
\newtheorem{proposition}{Proposition}
\newtheorem{corollary}{Corollary}

\title{\bf Sparse phase retrieval\\ via group-sparse optimization}

\author{F. Lauer$^1$ and H. Ohlsson$^{2,3}$
\medskip\\
$^1$ \small LORIA, Universit\'e de Lorraine, CNRS, Inria, France\\
$^2$ \small Dept. of Electrical Engineering and Computer Sciences,
University of California, Berkeley, USA\\
$^3$ \small Dept. of Electrical Engineering, Link\"oping University, Sweden
}

\begin{document}

\maketitle

\begin{abstract}
This paper deals with sparse phase retrieval, i.e., the problem of estimating a vector from quadratic measurements under the assumption that few components are nonzero. In particular, we consider the problem of finding the sparsest vector consistent with the measurements and reformulate it as a group-sparse optimization problem with linear constraints. Then, we analyze the convex relaxation of the latter based on the minimization of a block $\ell_1$-norm and show various exact recovery and stability results in the real and complex cases. Invariance to circular shifts and reflections are also discussed for real vectors measured via complex matrices.
\end{abstract}

\section{Introduction}

The problem of recovering a signal from quadratic measurements is known as phase retrieval. A typical case with many applications, for instance in optics \cite{Walther63} or crystallography \cite{Harrison93}, is when the measurements correspond to the (squared) magnitude of the Fourier transform of the signal. Here, we consider the more general setting of arbitrary quadratic measurements, $y_i= \g x^H \g Q_i\g x$, $i=1,\dots, N$, while focusing on the case where the signal $\g x\in\C^n$ is assumed to be sparse, i.e., with few nonzero entries in $\g x$. As in compressive sensing \cite{Donoho06,Candes06}, which deals with the recovery of sparse signals from linear measurements, the sparsity prior reduces the number of measurements required to recover the signal.

\paragraph{Related work.} 
Seminal works on phase retrieval \cite{Kohler72,Gerchberg72,Gonsalves76,Fienup82} did not consider the sparsity prior. Though these methods were able to incorporate prior information on the support, they were not designed to estimate the support with limited information on its size. 
More recently, matrix lifting techniques were developed in \cite{Candes13,Candes13b,Waldspurger12} for phase retrieval and more particularly for the case of sparse signals in \cite{Shechtman11,Jaganathan12,Ohlsson13quadratic}. The basic idea is to apply a change of variable resulting in linearized measurement equations with a rank-1 constraint on the new matrix variable $\g X = \g x\g x^H$. Then, the rank-1 constraint is relaxed to the problem of minimizing the rank of the matrix, which is further relaxed to the minimization of the nuclear norm. In these methods, the sparsity prior is typically incorporated as the minimization of an $\ell_1$-norm, which induces a trade-off between the satisfaction of the rank-1 constraint  and the sparsity of the solution. 
Other methods focusing on sparsity are typically iterative, like Fienup-type methods \cite{Mukherjee12} using alternate projections or the GESPAR method \cite{Shechtman13} implementing a local search strategy with a bi-directional greedy algorithm. These iterative methods usually come without recovery guarantees. 

Finally, note that sparse phase retrieval also enters the more general framework of nonlinear compressed sensing, as investigated in \cite{Ohlsson13nlbp} for analytic functions computing the measurements, in \cite{Ehler13} for quasi-linear functions and in \cite{Lauer13b} for polynomials.

\paragraph{Contribution. }
We propose a convex approach to the sparse phase retrieval problem. This approach relies on two main steps: the linearization of the constraints inducing a group-sparse structure on the variables and a convex relaxation of the group-sparse optimization problem enforcing this structure. 
More precisely, the linearization is based on the Veronese map lifting the signal to a higher dimensional space. This map is invariant to a global sign change and preserves sparsity in the sense that the lifted signal has a group-sparsity similar to the sparsity of the original signal. Thus, the proposed method amounts to estimating a group-sparse signal satisfying linear constraints, from which the original signal can be recovered. This estimation relies on a convex relaxation of the group-sparse problem based on a sum of norms, or block $\ell_1$-norm. Thus, while methods based on matrix lifting lead to semi-definite programming problems, the proposed approach yields a more amenable second-order cone programming formulation. 
This formulation is also easily extended to deal with noisy measurements. 
In addition to these algorithmic benefits, we derive exact recovery conditions in the noiseless case and stable recovery guarantees in the presence of noise. 

Note that the approach taken here is similar in spirit to the one derived in \cite{Lauer13b} for the more general problem of finding sparse solutions of polynomial systems of equations. However, the analysis in \cite{Lauer13b} is limited to the real case and does not apply to polynomials without linear terms as the ones found in phase retrieval. 

\paragraph{Paper organization.}
For the sake of clarity, we first detail in Sect.~\ref{sec:real} the proposed method in the real case before extending it in Sect.~\ref{sec:complex} to the complex case. The effect of noise and stability results are discussed in Sect.~\ref{sec:noise}. The case of real signals measured via complex vectors is considered in Sect.~\ref{sec:complexreal} which also deals with the invariance of the measurements to circular shifts and reflections. Finally, Section~\ref{sec:exp} tests the proposed methods in numerical experiments.

\paragraph{Notations.} 
Matrices are written with bold uppercase letters and vectors in bold lowercase letters, except for the $i$th column $\g A_i$ of a matrix $\g A$. The notation $(\g A)_{i,j}$ denotes the element at the $i$th row and $j$th column of a matrix $\g A$. $\Re(\cdot)$ and $\Im(\cdot)$ denote the real and imaginary parts of a complex number, vector or matrix, and $\i$ the imaginary unit. 
The superscripts $^T$ and $^H$ denote the transpose and conjugate transpose, respectively, i.e., $\g z^H = \overline{\g z}^T$. $\|\cdot\|_p$ denotes the $\ell_p$-norm in $\R^n$, while $\|\cdot\|$ denotes the norm in $\C^n$ induced by the inner product as $\|\g z\| = \sqrt{\g z^H\g z}$. The $\ell_0$-pseudo-norm of a real or complex vector $\g x$ of dimension $n$ is defined as $\|\g x\|_0 = \left|\{j\in\{1,\dots, n\} : x_j \neq 0\}\right|$ and denotes the number of nonzero components $x_j$. We also define the $\ell_0$-pseudo-norm of a vector-valued sequence $\{\g u_i\}_{i=1}^N$ as $\|\{\g u_i\}_{i=1}^N\|_0 = \left|\{i\in\{1,\dots, N\} : \g u_i \neq \g 0\}\right|$. 

\section{The real case} 
\label{sec:real}

We write the sparse phase retrieval problem as
\begin{align}\label{eq:P0}
	\min_{\g x\in\R^n} \ & \|\g x\|_0 \\
	\mbox{s.t.}\ & y_i = (\g q_i^T \g x)^2 = \g x^T\g Q_i \g x,\quad i=1,\dots,N, \nonumber
\end{align}
where $y_i\in \R$ are the measurements and $\g Q_i= \g q_i\g q_i^T  \in \R^{n\times n}$. 
Due to symmetry, solutions to \eqref{eq:P0} are defined up to their sign and the goal is to obtain an estimate $\hat{\g x} = \pm\g x_0$, for $\g x_0$ in the solution set of \eqref{eq:P0}. In particular, we are interested in the case where~\eqref{eq:P0} has a unique pair of solutions $\{\g x_0, - \g x_0\}$, while conditions ensuring such a uniqueness are discussed in \cite{Ohlsson13c,Ranieri13}.

The proposed method relies on two subsequent relaxations. While the first one linearizes the constraints, the second one convexifies the objective function. 

\subsection{First level of relaxation} 

Let the Veronese map of degree 2, $\nu : \R^n \rightarrow \R^M$, be defined by 
$$
	\nu(\g x) = [x_1^2,\ x_1x_2,\ \dots,\ x_2^2,\ x_2 x_3,\ \dots,\ x_{n-1}^2,\ x_{n-1} x_n,\ x_n^2]^T ,
$$
and the subscript $_{ij}$ denote the index of its component equal to $x_ix_j$, i.e.,
\begin{equation}\label{eq:ij}
	_{ij} = \sum_{k=1}^{\min \{i,j\} - 1} (n-k + 1) + |j-i| + 1 = \sum_{k=1}^{\min \{i,j\} - 1} (n-k) + \min \{i,j\} + |j-i|. 
\end{equation}
This notation is symmetric, i.e., $_{ij}$ and $_{ji}$ denote the same index, and will be used throughtout the paper to index the components of vectors of $\R^M$ or $\C^M$. 

The constraints of the phase retrieval problem~\eqref{eq:P0} can be rewritten as 
$$
 \quad\g A \nu(\g x) = \g y
$$
with $\g A \in\R^{N\times M}$ and $M=\begin{pmatrix}n+1\\2\end{pmatrix} = n(n+1)/2$. 

Let $\g W_j$ be an $n\times M$-binary matrix such that\footnote{More precisely, the entries of $\g W_j$ are given by $(\g W_j)_{k,l} = \delta_{l,jk}$, $l=1,\dots,M$, $k=1,\dots,n$, where $\delta$ is the Kronecker delta and $_{jk}$ is an index as in~\eqref{eq:ij}.} $\g W_j \nu(\g x) = x_j\g x$, i.e., $\g W_j \nu(\g x)$ is the vector of $n$ entries corresponding to the monomials in $\nu(\g x)$ including $x_j$. Then, we have
$$
	x_j = 0\ \Leftrightarrow\ \g W_j \nu (\g x) = \g 0 
$$
and the objective function in~\eqref{eq:P0} can be written as
\begin{equation}\label{eq:sparsenu}
	\|\g x\|_0 = \|\{\g W_j \nu(\g x)\}_{j=1}^n\|_0.
\end{equation}
Thus, \eqref{eq:P0} can be reformulated as the nonlinear group-sparse optimization problem
\begin{align}\label{eq:P0group}
	\min_{\g x\in\R^n} \ & \|\{\g W_j \nu (\g x)\}_{j=1}^n\|_0 \\
	\mbox{s.t.}\ &  \g A\nu (\g x) = \g y. \nonumber
\end{align}

Note that $\nu(\g x) = \nu(-\g x)$, but that for $\g x\neq \pm \g x_0$, $\nu(\g x)\neq \nu(\g x_0)$. Thus, the problem can be posed as the one of recovering the value of $\nu(\g x_0)$, from which $\g x_0$ can be inferred up to its sign. 

To estimate $ \nu(\g x_0)$, we relax \eqref{eq:P0group} to 
\begin{align}\label{eq:P0groupv}
	\min_{\g v\in\R^M} \ & \|\{\g W_j \g v\}_{j=1}^n\|_0 \\
	\mbox{s.t.}\ &  \g A\g v = \g y  \nonumber\\	
			& v_{jj} \geq 0,\ j=1,\dots,n, \nonumber
\end{align}
where the variables in $\g v$ estimating the components of $\nu(\g x)$ are not constrained to be interdepent monomials of $n$ base variables, but the last constraints in~\eqref{eq:P0groupv} nonetheless ensure that the $v_{jj}$'s estimating the $x_j^2$'s are positive. 

\subsection{Convex relaxation}

Problem~\eqref{eq:P0groupv} is a (linear) group-sparse optimization problem with highly overlapping groups. While groupwise-greedy algorithms, such as the one proposed in \cite{Lauer13b}, can be applied, their analysis is not available for the case of overlapping groups. Therefore, here, we consider the convex relaxation approach which aims at solving \eqref{eq:P0groupv} via the following surrogate formulation: 
\begin{align}\label{eq:P1group}
	\hat{\g v} = \arg\min_{\g v\in\R^M} \ & \sum_{j=1}^n \|\g W_j  \g W \g v\|_2 \\
	\mbox{s.t.}\ &  \g A\g v = \g y \nonumber\\	
			& v_{jj} \geq 0,\ j=1,\dots,n ,\nonumber
\end{align}
where we introduced the diagonal matrix $\g W$ of precompensating weights $(\g W)_{i,i} = w_i = \|\g A_i\|_2$, and which can be solved efficiently by off-the-shelf Second-Order Cone Programming (SOCP) solvers. 
Then, we easily obtain an estimate of $\g x_0$ from the estimate $\hat{\g v}$ of $\nu(\g x)$ as $\hat{\g x} = \nu^{-1}(\hat{\g v})$, where the inverse mapping $\nu^{-1}$ is defined as 
$$
	\nu^{-1}(\g v) = \begin{cases}
		\displaystyle{\frac{1}{ \sqrt{v_{ii}} }\left[v_{1i},\ v_{2i},\ \dots,\ v_{ni} \right]^T},\quad  \mbox{if } i > 0 \mbox{ and } \displaystyle{\frac{v_{ji}^2}{ v_{ii} } = v_{jj}},\ \forall j\in\{1,\dots,n\} \\
	\g 0,\ \mbox{otherwise} ,
	\end{cases}
$$
where
$$
	i= \begin{cases}
		\displaystyle{\min_{j\in\{1,\dots,n\}} j,\ \mbox{s.t. } v_{jj} > 0,\ \mbox{if } \exists j\ \mbox{such that }  \ v_{jj}> 0}\\ 
		 0,\quad \mbox{otherwise}  .
	\end{cases}
$$
In the above, the first nonzero entry (of index $i$) is assumed to be positive to fix the signs and make $\nu^{-1}$ injective. 
This definition also ensures that
$$
	\nu^{-1}(\nu(\g x)) = \pm \g x
$$
and that 
\begin{equation}\label{eq:sparsenu_1}
	\|\nu^{-1}(\g v)\|_0 \leq \|\{\g W_j\g v\}_{j=1}^n\|_0 
\end{equation}
since, for $\nu^{-1}(\g v)\neq \g 0$, $\g W_j \g v = 0 \Rightarrow v_{ji}=0 \Rightarrow \left(\nu^{-1}(\g v) \right)_{j} = 0$. 

\subsection{Analysis}

We now turn to theoretical guarantees offered by the proposed approach. 
First, the following theorem provides the rationale for tackling the sparse phase retrieval problem via the group-sparse optimization formulation \eqref{eq:P0groupv}. 

\begin{theorem}\label{thm:grouppoly}
	If the solution $\g v^*$ to \eqref{eq:P0groupv} is unique and yields $\g x^* = \nu^{-1}(\g v^*) \neq \g 0$ such that $y_i =  (\g x^*)^T\g Q_i \g x^*$, $i=1,\dots,N$, then $\{\g x^*, -\g x^*\}$ is the unique pair of solutions of \eqref{eq:P0}.
\end{theorem}
\begin{proof}
Assume there is an $\g x_0\neq \pm \g x^*$ satisfying the constraints of \eqref{eq:P0} and at least as sparse as $\g x^*$. Then, $\g A\nu(\g x_0) = \g y$ and, by using~\eqref{eq:sparsenu} and~\eqref{eq:sparsenu_1}, 
$$
	\|\{\g W_j \nu(\g x_0)\}_{j=1}^n\|_0 = \|\g x_0\|_0 \leq \|\g x^*\|_0\leq \|\{\g W_j\g v^*\}_{j=1}^n\|_0 ,
$$
which contradicts the fact that $\g v^*$ is the unique solution to \eqref{eq:P0group} unless $\nu(\g x_0)=\g v^*$. But since $\g x_0 \neq \pm \g x^*$, we have  $x_{0j}^2 \neq (x_j^*)^2$ for some $j\in\{1,\dots,n\}$, which implies $\left(\nu(\g x_0)\right)_{jj}\neq \left(\nu(\g x^*)\right)_{jj} = \left(\nu(\nu^{-1}(\g v^*))\right)_{jj}$. 
Therefore, by using Lemma~\ref{lem:nudiffreal} in Appendix~\ref{sec:lemmas} with the assumption $\g x^*= \nu^{-1}(\g v^*) \neq \g 0$, there cannot be such an $\g x_0$.
\end{proof}

The following results regarding the convex formulation~\eqref{eq:P1group} are based on the notion of mutual coherence. 
\begin{definition}\label{def:mu} The {\em mutual coherence} of a matrix $\g A = [\g A_1,\dots, \g A_M] \in\R^{N\times M}$ is 
$$
	\mu(\g A) = \max_{1\leq i< j \leq M} \frac{|\g A_i^T \g A_j|}{\|\g A_i\|_2 \|\g A_j\|_2} .
$$
\end{definition}

With this definition, we can state an exact recovery result (proof given in Appendix~\ref{proof:groupsparse}). 

\begin{theorem}\label{thm:groupsparse}
Let $\g x_0$ be such that $y_i = \g x_0^T\g Q_i \g x_0$, $i=1,\dots,N$, and $\g v_0=\nu(\g x_0)$. If the condition
$$
	 \|\g x_0\|_0 <  \frac{1}{2\sqrt{n}}\sqrt{1+ \frac{1}{\mu^2(\g A)}} 
$$
holds, then $\g v_0$ is the unique solution to \eqref{eq:P1group}.
\end{theorem}

\begin{corollary}\label{cor:uniqueness}
Let $\g x_0$ be a feasible point of \eqref{eq:P0}. If the condition
$$
	 \|\g x_0\|_0 <  \frac{1}{2\sqrt{n}}\sqrt{1+ \frac{1}{\mu^2(\g A)}} 
$$
holds, then $\{\g x_0, -\g x_0\}$ is the unique pair of solutions to the minimization problem~\eqref{eq:P0} and they can be computed as $\g x_0 = \pm \nu^{-1}(\hat{\g v})$ with $\hat{\g v}$ the solution to \eqref{eq:P1group}.
\end{corollary}
\begin{proof}
Assume there exists another solution $\g x_1\neq \pm\g x_0$ to \eqref{eq:P0}, and thus with $\|\g x_1\|_0 \leq \|\g x_0\|_0$. Then, Theorem \ref{thm:groupsparse} implies that both $\nu(\g x_1)$ and $\nu(\g x_0)$ are {\em unique} solutions to \eqref{eq:P1group} and thus that $\nu(\g x_1)=  \nu(\g x_0) = \hat{\g v}$. But this contradicts the definition of the mapping $\nu$ implying $\nu(\g x_1)\neq \nu(\g x_0)$ whenever $\g x_1\neq\pm\g x_0$. Therefore the assumption $\g x_1\neq\pm\g x_0$ cannot hold and $\{\g x_0, -\g x_0\}$ is the unique pair of solutions to \eqref{eq:P0}, while $\nu^{-1}(\hat{\g v}) = \nu^{-1}(\nu(\g x_0)) = \pm\g x_0$.
\end{proof}

\section{The complex case} 
\label{sec:complex}

Consider now the problem in complex domain:
\begin{align}\label{eq:P0complex}
	\min_{\g x\in\C^n} \ & \|\g x\|_0 \\
	\mbox{s.t.}\ & y_i =  |\g q_i^H \g x|^2,\quad i=1,\dots,N, \nonumber
\end{align}
where $y_i\in \R$ and $\g q_i \in \C^n$.

The equations in the problem above are invariant to  multiplication by a unit complex scalar $z$ with $|z|=1$. Thus, there are sets of solutions of the form $T(\g x_0) = \{\g x \in \C^n : \g x = z \g x_0,\ z\in \C,\ |z| = 1\}$, and the goal is to obtain an estimate $\hat{\g x} \in T(\g x_0)$, from which $T(\g x_0)=T(\hat{\g x})$ can be inferred due to the property of the invariance set:
\begin{equation}\label{eq:Tequality}
	\forall \g x\in T(\g x_0),\quad T(\g x) = T(\g x_0) ,
\end{equation}
which can be proved as follows. Let $\g x = z\g x_0$ with $|z|=1$, then $T(\g x) = \{\g a \in \C^n : \g a = b \g x,\ b\in \C,\ |b| = 1\} =  \{\g a \in \C^n : \g a = b z \g x_0,\ b\in \C,\ |b| = 1\} = \{\g a \in \C^n : \g a = c \g x_0,\ c\in \C,\ |c| = 1\} = T(\g x_0)$.

\subsection{First level of relaxation}
\label{sec:firstrelax}

As for the real case, the linearization of the equations will use the Veronese map, which we redefine for complex vectors as follows. 

\begin{definition}[Complex Veronese map]\label{def:veronesecomplex}
The complex Veronese map $\nu : \C^n\rightarrow \C^M$ is defined by
$$
	\nu(\g x) = [x_1\overline{x_1},\ x_1\overline{x_2},\ \dots,\ x_2\overline{x_2},\ x_2\overline{x_3},\ \dots,\ x_{n-1}\overline{x_{n-1}},\ x_{n-1}\overline{x_n},\ x_n\overline{x_n} ]^T ,
$$
for which the subscript $_{ij}$, defined as in~\eqref{eq:ij}, denotes the component index such that $\left(\nu(\g x)\right)_{ij}$ equals either $x_i\overline{x_j}$ or $x_j\overline{x_i}$ (note that, for all pairs $(i,j)$, there is exactly one such component). 
\end{definition}

Note that $\nu(\g x) = \nu(\g x^\prime)$ for all $\g x^\prime \in T(\g x)$, since for $1\leq i\leq j\leq n$, $x_i^\prime\overline{x_j^\prime} = z x_j\overline{z x_j} = |z|^2 x_i\overline{x_j} = x_i\overline{x_j}$, but that $\g x^\prime \notin T(\g x)$ implies $\nu(\g x) \neq \nu(\g x^\prime)$, since $\g x^\prime \notin T(\g x)\Rightarrow |x_j^\prime| \neq |x_j| \Rightarrow x_j^\prime\overline{x_j^\prime} \neq x_j\overline{x_j}$. 

Then, we define the inverse mapping as follows. 
\begin{definition}[Inverse complex Veronese map]\label{def:nuinv}
The inverse complex Veronese map, $\nu^{-1} : \C^M\rightarrow \C^n$, is defined by 
$$
	\nu^{-1}(\g v) = \begin{cases}
		\displaystyle{\frac{1}{ \sqrt{v_{ii}} }\left[\overline{v_{1i}},\ \overline{v_{2i}},\ \dots,\ \overline{v_{ni} }\right]^T},\quad \mbox{if } i > 0 \mbox{ and } \displaystyle{\frac{|v_{ji}|^2}{ v_{ii} } = v_{jj}},\ \forall j\in\{1,\dots,n\} \\
	\g 0,\ \mbox{otherwise} ,
	\end{cases}
$$
where
$$
	i= \begin{cases}
		\displaystyle{\min_{j\in\{1,\dots,n\}} j,\ \mbox{s.t. } \Re(v_{jj} ) > 0,\ \Im(v_{jj}) = 0},\ \mbox{if } \exists j\ \mbox{such that }  \ \Re(v_{jj})> 0, \Im(v_{jj}) = 0\\ 
		 0,\quad \mbox{otherwise}  .
	\end{cases}
$$
In particular, we have $\nu^{-1}(\nu(\g x)) \in T(\g x)$. 
\end{definition}
Note that the square root acts on a real and positive number $v_{ii}$ and is thus also a real positive number. This implies $x_i = \sqrt{v_{ii}}\in\R^+$ and that we arbitrarily set $\Im(x_i) = 0$ to fix the value of $z$ in the equation $\nu^{-1}(\nu(\g x)) = z \g x$, for a complex number $z$ with $|z| = 1$, and thus make $\nu^{-1}$ injective.

With these definitions at hand, the equations in~\eqref{eq:P0complex} are reformulated via Lemma~\ref{lem:complexquadratic} (all Lemmas are given in Appendix~\ref{sec:lemmas}) as follows:
$$
	y_i =  \g x^H \g q_i \g q_i^H \g x = 2 \Re\left(\nu(\g q_i)^H\nu(\g x)\right) - \sum_{j=1}^n  \left(\nu(\g q_i)\right)_{jj}  \left(\nu(\g x)\right)_{jj}, \quad i=1,\dots,N . 
$$
Define the vectors $\g a_i\in\C^M$, $i=1,\dots,N$, with components given by $(\g a_i)_{jk} = 2\left(\nu(\g q_i)\right)_{jk}$ for $1\leq j < k \leq n$ and $(\g a_i)_{jj} = \left(\nu(\g q_i)\right)_{jj}$, $j=1\dots,n$. Then, the equations above, linear wrt. to $\nu(\g x)$, can be rewritten as $y_i = \Re\left(\g a_i^H\nu(\g x)\right)$, $i=1,\dots,N$. 

Additionally define the binary matrices $\g W_j$ such that the vector $\g W_j\nu(\g x)$ contains all the monomials of $\nu(\g x)$ including either $x_j$ or $\overline{x_j}$: 
$$
	\g W_j\nu(\g x)= [x_1 \overline{x_j},\ x_2 \overline{x_j},\ \dots,\ x_{j-1} \overline{x_j},\ x_j \overline{x_j}, \ x_j \overline{x_{j+1}},\ \dots,\ x_j \overline{x_{n}}]^T \in \C^n.
$$
Then, we have 
\begin{equation}\label{eq:sparsenucomplex}
	\|\g x\|_0 = \|\{\g W_j \nu(\g x)\}_{j=1}^n\|_0
\end{equation}
and problem~\eqref{eq:P0complex} can be rewritten as the nonlinear group-sparse optimization program
\begin{align}\label{eq:P0nucomplex}
	\min_{\g x\in\C^n} \ & \|\{\g W_j \nu(\g x)\}_{j=1}^n\|_0 \\
	\mbox{s.t.}\ & \g y = \Re\left(\g A \nu(\g x)\right) , \nonumber	
\end{align}
where $\g A = [\g a_1, \dots, \g a_N]^H$. 

Next, we relax this formulation by substituting $\g v \in\C^M$ for $\nu(\g x)$, which yields
\begin{align}\label{eq:P0vcomplex}
	\min_{\g v\in\C^M} \ & \|\{\g W_j\g v\}_{j=1}^n\|_0 \\
	\mbox{s.t.}\ & \g y = \Re\left(\g A \g v\right) \nonumber\\
	& v_{jj} \in\R^+,\ j=1,\dots,n ,\nonumber
\end{align}
where the last constraints ensure that the $v_{jj}$'s estimating the modulus of the base variables are positive real numbers. 

The following theorem shows that this relaxation can be used as a proxy to solve the original problem. 
\begin{theorem}\label{thm:grouppolycomplex}
	If the solution $\g v^*$ to \eqref{eq:P0vcomplex} is unique and yields $\g x^* = \nu^{-1}(\g v^*) \neq \g 0$ such that $y_i =  |\g q_i^H\g x^*|^2$, $i=1,\dots,N$, then $T(\g x^*)$ is the unique set of solutions of \eqref{eq:P0complex}. 
\end{theorem}
\begin{proof}
Assume there is an $\g x_0\neq T(\g x^*)$ satisfying the constraints of \eqref{eq:P0complex} and at least as sparse as $\g x^*$. Then 
$\g y = \Re\left(\g A \nu(\g x_0)\right)$, 
and, by using~\eqref{eq:sparsenucomplex} and Lemma~\ref{lem:sparsitynuinv}, 
$$
	\|\{\g W_j \nu(\g x_0)\}_{j=1}^n\|_0 = \|\g x_0\|_0 \leq \|\g x^*\|_0\leq \|\{\g W_j\g v^*\}_{j=1}^n\|_0 ,
$$
which contradicts the fact that $\g v^*$ is the unique solution to \eqref{eq:P0vcomplex} unless $\nu(\g x_0)=\g v^*$. But since $\g x_0 \notin T(\g x^*)$, we have $|x_{0j}| \neq |x_j^*|$ and thus $x_{0j}\overline{x_{0j}} \neq x_j^*\overline{x_j^*}$ for some $j\in\{1,\dots,n\}$, which implies $\left(\nu(\g x_0)\right)_{jj}\neq \left(\nu(\g x^*)\right)_{jj} = \left(\nu(\nu^{-1}(\g v^*))\right)_{jj}$. 
Therefore, by using Lemma~\ref{lem:nudiff} with the assumption $\g x^*= \nu^{-1}(\g v^*) \neq \g 0$, there cannot be such an $\g x_0$.
\end{proof}

\subsection{Convex relaxation}

We now introduce a second level of relaxation by replacing the $\ell_0$-pseudo norm by a block-$\ell_1$ norm. 
This leads to a convex relaxation in the form of a SOCP:
\begin{align}\label{eq:P1vcomplex}
	\min_{\g v\in\C^M} \ & \sum_{j=1}^n \|\g W_j^R \Re( \g v) + \i \g W_j^I \Im( \g v)\|\\
	\mbox{s.t.}\ & \g y = \Re\left(\g A\g v\right)  \nonumber\\
	& v_{jj} \in\R^+,\ j=1,\dots,n ,\nonumber
\end{align}
where $\g W_j^R = \g W_j \g W^R$ and $\g W_j^I = \g W_j \g W^I$ with $\tilde{ \g W} =  \begin{pmatrix}\g W^R & \g 0\\\g 0 & \g W^I\end{pmatrix}$ a diagonal matrix of precompensating weights given by $( \g W^R)_{i,i}= \|\Re(\g A_i)\|_2$ and $(\g W^I)_{i,i}= \|\Im(\g A_i)\|_2$.

\begin{theorem}\label{thm:groupsparsecomplex}
Let $\g x_0$ be such that $y_i = \g x_0^H\g q_i \g q_i^H\g x_0$, $i=1,\dots,N$, $\g v_0=\nu(\g x_0)$ and $\tilde{\g A} = [\Re(\g A),\ -\Im(\g A)]$. If the condition
$$
	 \|\g x_0\|_0 <  \frac{1}{2\sqrt{2n}}\sqrt{1+ \frac{1}{\mu^2(\tilde{\g A})}} 
$$
holds, then $\g v_0$ is the unique solution to \eqref{eq:P1vcomplex}.
\end{theorem}
\begin{proof}
The vector $\g v_0$ is the unique solution to \eqref{eq:P1vcomplex} if  the inequality
$$
	\sum_{j=1}^n \|\g W_j^R \Re(\g v_0 + \g \delta) + \i \g W_j^I \Im(\g v_0 + \g \delta)\| > \sum_{j=1}^n \|\g W_j^R \Re( \g v_0) + \i \g W_j^I \Im( \g v_0)\|
$$
holds for all $\g \delta\in\C^M$ such that $\Re(\g A(\g v_0 + \g \delta) ) = \g y$, which implies the constraint $ \Re\left(\g A\g \delta\right) = \g 0$ on $\g\delta$.
The inequality above can be rewritten as
$$
	\sum_{j\in I_0} \|\g W_j^R \Re( \g \delta) + \i \g W_j^I \Im(\g \delta)\| + \sum_{j\notin I_0} \|\g W_j^R \Re(\g v_0 + \g \delta) + \i \g W_j^I \Im(\g v_0 + \g \delta)\| -  \|\g W_j^R \Re( \g v_0) + \i \g W_j^I \Im( \g v_0)\| > 0 ,
$$
where $I_0=\{j\in\{1,\dots,n\} : \g W_j^R \Re(\g v_0) = \g 0 \ \wedge\ \g W_j^I \Im(\g v_0) = \g 0\}$.  
By the triangle inequality, $\|\g a + \g b\| - \|\g a\| \geq -\|\g b\|$ with $\g a = \g W_j^R \Re( \g v_0) + \i \g W_j^I \Im( \g v_0)$, this condition is met if
$$
	\sum_{j\in I_0} \|\g W_j^R \Re( \g \delta) + \i \g W_j^I \Im(\g \delta)\|  -\sum_{j\notin I_0} \|\g W_j^R \Re(\g \delta) + \i \g W_j^I \Im(\g \delta)\| > 0
$$
or
\begin{equation}\label{eq:proof1cplx}
	\sum_{j=1}^n \|\g W_j^R \Re( \g \delta) + \i \g W_j^I \Im(\g \delta)\|  - 2\sum_{j\notin I_0} \|\g W_j^R \Re(\g \delta) + \i \g W_j^I \Im(\g \delta)\| > 0 .
\end{equation}

By defining $G_j$ as the set of indexes corresponding to nonzero columns of $\g W_j$, Lemma~\ref{lem:deltacomplex} yields
\begin{align*}
	\|\g W_j^R \Re( \g \delta) + \i \g W_j^I \Im(\g \delta)\|^2 &= \sum_{i\in G_j} (w_i^R)^2 \Re(\delta_i)^2 + (w_i^I)^2 \Im(\delta_i)^2   
	\\& \leq 2 n\frac{\mu^2(\tilde{\g A})}{1+\mu^2(\tilde{\g A})} \|\g W^R \Re(\g \delta) + \i \g W^I \Im(\g \delta) \|^2  .
\end{align*}
Due to the fact that $\bigcup_{k\in\{1,\dots,n\}} G_k = \{1,\dots,M\}$, we also have
\begin{align*}
	\|\g W^R \Re(\g \delta) + \i \g W^I \Im(\g \delta) \|^2  &= \sum_{i=1}^M (w_i^R)^2 \Re(\delta_i)^2 + (w_i^I)^2 \Im(\delta_i)^2 \\
	&\leq  \sum_{k=1}^n \sum_{i\in G_k} (w_i^R)^2 \Re(\delta_i)^2 + (w_i^I)^2 \Im(\delta_i)^2 \\
	&= \sum_{k=1}^n \|\g W_k^R \Re( \g \delta) + \i \g W_k^I \Im(\g \delta)\|^2 \\
	& \leq \left(\sum_{k=1}^n \|\g W_k^R \Re( \g \delta) + \i \g W_k^I \Im(\g \delta)\| \right)^2 ,
\end{align*}
which then leads to
\begin{align*}
	\|\g W_j^R \Re( \g \delta) + \i \g W_j^I \Im(\g \delta)\|^2
	& \leq 2 n\frac{\mu^2(\tilde{\g A})}{1+\mu^2(\tilde{\g A})} \left(\sum_{k=1}^n \|\g W_k^R \Re( \g \delta) + \i \g W_k^I \Im(\g \delta)\| \right)^2.
\end{align*}
Introducing this result in \eqref{eq:proof1cplx} gives the condition
$$
	\sum_{j=1}^n \|\g W_j^R \Re( \g \delta) + \i \g W_j^I \Im(\g \delta)\|  - 2 (n - |I_0|) \frac{\mu(\tilde{\g A}) \sqrt{2n}}{\sqrt{1+\mu^2(\tilde{\g A})}} \sum_{k=1}^n \|\g W_k^R \Re( \g \delta) + \i \g W_k^I \Im(\g \delta)\|  > 0 .
$$
Finally, given that $|I_0| = n - \|\g x_0\|_0$, this yields
\begin{equation}\label{eq:conduniqueP1}
	\sum_{j=1}^n \|\g W_j^R \Re( \g \delta) + \i \g W_j^I \Im(\g \delta)\|  - 2\|\g x_0\|_0 \frac{\mu(\tilde{\g A}) \sqrt{2n}}{\sqrt{1+\mu^2(\tilde{\g A})}} \sum_{k=1}^n \|\g W_k^R \Re( \g \delta) + \i \g W_k^I \Im(\g \delta)\|  > 0 .
\end{equation}
or, for $\g\delta\neq \g 0$,
$$
	 \|\g x_0\|_0 < \frac{\sqrt{1+\mu^2(\tilde{\g A})}}{2\mu(\tilde{\g A}) \sqrt{2n} } ,
$$
which can be rewritten as in the statement of the Theorem. 
\end{proof}

\begin{corollary}\label{cor:uniquenesscomplex}
Let $\g x_0$ be a feasible point of \eqref{eq:P0complex}. If the condition
$$
	 \|\g x_0\|_0 <  \frac{1}{2\sqrt{2n}}\sqrt{1+ \frac{1}{\mu^2(\tilde{\g A})}} 
$$
holds, then $T(\g x_0)$ is the unique set of solutions to the minimization problem~\eqref{eq:P0complex} and it can be computed as $T(\g x_0) = T(\nu^{-1}(\hat{\g v}))$ with $\hat{\g v}$ the solution to \eqref{eq:P1vcomplex}.
\end{corollary}
\begin{proof}
Assume there exists another solution $\g x_1\notin T(\g x_0)$ to \eqref{eq:P0complex}, and thus with $\|\g x_1\|_0 \leq \|\g x_0\|_0$. Then, Theorem \ref{thm:groupsparsecomplex} implies that both $\nu(\g x_1)$ and $\nu(\g x_0)$ are {\em unique} solutions to \eqref{eq:P1vcomplex} and thus that $\nu(\g x_1)=  \nu(\g x_0) = \hat{\g v}$. But this contradicts Definition~\ref{def:veronesecomplex} implying $\nu(\g x_1)\neq \nu(\g x_0)$ whenever $\g x_1\notin T(\g x_0)$. Therefore the assumption $\g x_1\notin T(\g x_0)$ cannot hold and $T(\g x_0)$ is the unique set of solutions to \eqref{eq:P0complex}, while $\nu^{-1}(\hat{\g v}) = \nu^{-1}(\nu(\g x_0)) \in T(\g x_0)$, which, by using~\eqref{eq:Tequality}, implies $T(\g x_0) = T(\nu^{-1}(\hat{\g v}))$.
\end{proof}

\section{Stable recovery in the presence of noise}
\label{sec:noise}

Consider now the case where the measurements $\g y$ are perturbed by an additive noise $\g e\in\R^N$ of bounded $\ell_2$-norm, $\|\g e\|_2 \leq \varepsilon$. Then, the equations in~\eqref{eq:P0} and~\eqref{eq:P0complex} are of the form $y_i =  |\g q_i^H \g x|^2 + e_i$ with the noise terms $e_i$ to be estimated together with the sparse signal. Note that in this context, multiple solutions with different noise vectors can be valid. Thus, we aim at stability results bounding the error on the estimates by a function of $\varepsilon$ rather exact recovery ones. 
Details on the proposed method to achieve these goals are given below, first for real signals and then for complex ones. 

\subsection{Stability in the real case}

In the noisy case with real data, the problem that we need to solve becomes
\begin{align}\label{eq:P0noisereal}
	\min_{\g x\in\R^n, \g e\in\R^N} \ & \|\g x\|_0 \\
	\mbox{s.t.}\ & y_i =  |\g q_i^T \g x|^2 + e_i,\quad i=1,\dots,N, \nonumber\\
	& \|\g e\|_2 \leq \varepsilon, \nonumber
\end{align}
where $y_i\in \R$ and $\g q_i \in \R^n$.
Following the approach of Sect.~\ref{sec:real} leads to a convex relaxation in the form of the SOCP 
\begin{align}\label{eq:P1noisereal}
	\min_{\g v\in\R^M} \ & \sum_{j=1}^n  \|\g W_j \g W \g v\|_2\\
	\mbox{s.t.}\ & \| \g y - \g A\g v \|_2 \leq \varepsilon \nonumber\\
	& v_{jj} \geq 0,\ j=1,\dots,n ,\nonumber
\end{align}
for which we have the following stability result. 
\begin{theorem}\label{thm:stabilitygroupreal}
	Let $(\g x_0, \g e_0)$ denote a solution to \eqref{eq:P0noisereal}. If the inequality 
	\begin{equation}\label{eq:conditionP1groupnoisereal}
		\|\g x_0\|_0 < \frac{1}{2n^2(n+1)}\left(1+ \frac{1}{\mu(\g A)}\right) 
	\end{equation}
	holds, then  the solution $\hat{\g v}$ to 
	\eqref{eq:P1noisereal} must obey
	\begin{equation}\label{eq:stabilityxgroupWreal}
		\|\g W (\hat{\g v} - \nu(\g x_0)) \|_2^2 \leq \frac{4n \varepsilon^2}{1 - \mu(\g A)[2n^2(n+1)\|\g x_0\|_0 - 1]} .
	\end{equation}
	If, in addition, $\varepsilon = 0$, then $\hat{\g x} =\pm \g x_0$.
\end{theorem}
We omit the proof which is similar to the one of Theorem~6 in \cite{Lauer13b}, except for the last statement concluding on the stability of $\hat{\g x}$, and which closely follows the one for the complex case of Theorem~\ref{thm:stabilitygroup} to be detailed below.
 
\subsection{Stability in the complex case}

Consider now the complex variant of the problem perturbed by noise:
\begin{align}\label{eq:P0noise}
	\min_{\g x\in\C^n, \g e\in\R^N} \ & \|\g x\|_0 \\
	\mbox{s.t.}\ & y_i =  |\g q_i^H \g x|^2 + e_i,\quad i=1,\dots,N, \nonumber\\
	& \|\g e\|_2 \leq \varepsilon, \nonumber
\end{align}
where $y_i\in \R$ and $\g q_i \in \C^n$.
The solution to this problem can be approached via the convex relaxation
\begin{align}\label{eq:P1noise}
	\min_{\g v\in\C^M} \ & \sum_{j=1}^n  \|\g W_j^R \Re( \g v) + \i \g W_j^I \Im( \g v)\|\\
	\mbox{s.t.}\ & \| \g y - \Re\left(\g A\g v\right) \|_2 \leq \varepsilon \nonumber\\
	& v_{jj} \in\R^+,\ j=1,\dots,n .\nonumber
\end{align}
As for the real case, we have a stability result for the estimation of $\nu(\g x)$. 

\begin{theorem}\label{thm:stabilitygroup}
	Let $(\g x_0, \g e_0)$ denote a solution to \eqref{eq:P0noise}. If the inequality 
	\begin{equation}\label{eq:conditionP1groupnoise}
		\|\g x_0\|_0 < \frac{1}{2n^2(n+1)}\left(1+ \frac{1}{\mu(\tilde{\g A})}\right) 
	\end{equation}
	holds, then  the solution $\hat{\g v}$ to 
	\eqref{eq:P1noise} must obey
	\begin{equation}\label{eq:stabilityxgroupW}
		\|\g W^R \Re(\hat{\g v} - \nu(\g x_0)) + \i \g W^I \Im(\hat{\g v} - \nu(\g x_0))\|^2 \leq \frac{4n \varepsilon^2}{1 - \mu(\tilde{\g A})[2n^2(n+1)\|\g x_0\|_0 - 1]} .
	\end{equation}
	If, in addition, $\varepsilon = 0$, then $\hat{\g x} = \nu^{-1}(\hat{\g v}) \in T(\g x_0)$.
\end{theorem}
\begin{proof}
Assume \eqref{eq:P0noise} has a solution $(\g x_0,\g e_0)$. Let define  $\g v_0 = \nu(\g x_0)$ and $\g \delta = \hat{\g v} - \g v_0$. The proof follows a path similar to that of Theorem~3.1 in \cite{Donoho06noise}, which was adapted in \cite{Lauer13b} to the group-sparse setting and which is here further extended to the complex case.

Due to the definition of $\hat{\g v}$ as a minimizer of \eqref{eq:P1noise}, $\g \delta$ must satisfy either
$$
	\sum_{j=1}^n \|\g W_j^R \Re(\g v_0 + \g \delta) + \i \g W_j^I \Im(\g v_0 + \g \delta)\| < \sum_{j=1}^n \|\g W_j^R \Re( \g v_0) + \i \g W_j^I \Im( \g v_0)\|
$$
or $\g \delta = \g 0$, in which case the statement is obvious. 
The inequality above can be rewritten as
$$
\sum_{j\in I_0} \|\g W_j^R \Re( \g \delta) + \i \g W_j^I \Im(\g \delta)\| + \sum_{j\notin I_0} \|\g W_j^R \Re(\g v_0 + \g \delta) + \i \g W_j^I \Im(\g v_0 + \g \delta)\| -  \|\g W_j^R \Re( \g v_0) + \i \g W_j^I \Im( \g v_0)\| < 0 ,
$$
where $I_0=\{j\in\{1,\dots,n\} : \g W_j^R \Re(\g v_0) = \g 0 \ \wedge\ \g W_j^I \Im(\g v_0) = \g 0\}$.  
By the triangle inequality, $\|\g a + \g b\|- \|\g a\| \geq -\|\g b\|$, this implies 
\begin{equation}\label{eq:proofainfb}
	\sum_{j\in I_0}\|\g W_j^R \Re( \g \delta) + \i \g W_j^I \Im(\g \delta)\|  - \sum_{j\notin I_0}\|\g W_j^R \Re( \g \delta) + \i \g W_j^I \Im(\g \delta)\|  < 0.
\end{equation}
In addition, $\g \delta$ must satisfy the constraints in \eqref{eq:P1noise} as 
$$
	\|\Re(\g A(\g v_0 + \g \delta)) - \g y\|_2\leq \varepsilon ,
$$	
in which $\g y$ can be replaced by $\Re(\g A \g v_0 ) + \g e_0$, leading to 
$$
	\|\Re(\g A\g \delta ) - \g e_0\|_2\leq \varepsilon .
$$
Using $\|\g a\|_2 \leq \|\g a - \g b\|_2 + \|\g b\|_2$, this implies $\|\Re(\g A\g \delta)\|_2\leq 2 \varepsilon$, which further gives
\begin{align}
	4\varepsilon^2 &\geq \|\Re(\g A\g \delta)\|_2^2 = \|\tilde{\g A} \tilde{\g \delta}\|_2^2  = \|\tilde{\g A} \tilde{\g W}^{-1}\tilde{\g W}\tilde{\g \delta}\|_2^2 \nonumber
	= (\tilde{\g W} \tilde{\g \delta})^T \tilde{\g W}^{-1} \tilde{\g A}^T \tilde{\g A}\tilde{\g W}^{-1} (\tilde{\g W} \tilde{\g \delta}) \nonumber\\
	&= \|\tilde{\g W}\tilde{\g \delta}\|_2^2  + (\tilde{\g W} \tilde{\g \delta})^T ( \tilde{\g W}^{-1} \tilde{\g A}^T \tilde{\g A}\tilde{\g W}^{-1} - \g I) (\tilde{\g W} \tilde{\g \delta}) \nonumber\\
	&\geq \|\tilde{\g W}\tilde{\g \delta}\|_2^2  - \left| (\tilde{\g W} \tilde{\g \delta})^T ( \tilde{\g W}^{-1} \tilde{\g A}^T \tilde{\g A}\tilde{\g W}^{-1} - \g I) (\tilde{\g W} \tilde{\g \delta})\right|\nonumber\\
	&\geq \|\tilde{\g W}\tilde{\g \delta}\|_2^2  - | \tilde{\g W} \tilde{\g \delta}|^T | \tilde{\g W}^{-1} \tilde{\g A}^T \tilde{\g A}\tilde{\g W}^{-1} - \g I| |\tilde{\g W} \tilde{\g \delta}|\nonumber\\
	&\geq \|\tilde{\g W}\tilde{\g \delta}\|_2^2  - \mu(\tilde{\g A}) (\|\tilde{\g W}\tilde{\g \delta}\|_1^2 - \|\tilde{\g W}\tilde{\g \delta}\|_2^2 ) \nonumber\\
	&= ( 1 + \mu(\tilde{\g A}) )\|\tilde{\g W}\tilde{\g \delta}\|_2^2  - \mu(\tilde{\g A}) \|\tilde{\g W}\tilde{\g \delta}\|_1^2  
	\label{eq:proofstabilitygroup1}
\end{align}
where we used $\tilde{\g W}^{-1}\tilde{\g W} = \g I$ and the fact that the diagonal entries of $|\tilde{\g W}^{-1} \tilde{\g A}^T \tilde{\g A}\tilde{\g W}^{-1} - \g I|$ are zeros while off-diagonal entries are bounded from above by $\mu(\tilde{\g A})$. 

Due to $\g W_j^R \Re( \g \delta) + \i \g W_j^I \Im(\g \delta)$ being a vector with a subset of entries from $\g W^R \Re( \g \delta) + \i \g W^I \Im(\g \delta)$, we have $\|\tilde{\g W}\tilde{\g \delta}\|_2^2 = \|\g W^R \Re(\g \delta) + \i \g W^I \Im(\g \delta) \|^2  \geq  \|\g W_j^R \Re( \g \delta) + \i \g W_j^I \Im(\g \delta)\|^2$, $j=1,\dots,n$, and thus
\begin{equation}\label{eq:bounddeltaL2}
	 \|\tilde{\g W}\tilde{\g \delta}\|_2^2 \geq  \frac{1}{n} \sum_{j=1}^n  \|\g W_j^R \Re( \g \delta) + \i \g W_j^I \Im(\g \delta)\|^2 
	 . 
\end{equation}
Since the groups defined by the $\g W_j$'s overlap, $\|\tilde{\g W}\tilde{\g \delta}\|_2=\|\g W^R \Re( \g \delta) + \i \g W^I \Im(\g \delta)\| \leq \sum_{j=1}^n  \|\g W_j^R \Re( \g \delta) + \i \g W_j^I \Im(\g \delta)\|$, and the squared $\ell_1$-norm in~\eqref{eq:proofstabilitygroup1} can be bounded by 
\begin{equation}\label{eq:bounddeltaL1}
	\|\tilde{\g W}\tilde{\g \delta}\|_1^2  \leq M \|\tilde{\g W}\tilde{\g \delta}\|_2^2 \leq M\left(  \sum_{j=1}^n  \|\g W_j^R \Re( \g \delta) + \i \g W_j^I \Im(\g \delta)\|\right)^2 .
\end{equation}
Introducing the bounds \eqref{eq:bounddeltaL2}--\eqref{eq:bounddeltaL1} in \eqref{eq:proofstabilitygroup1} yields
\begin{equation}\label{eq:proofstabilitygroup2}
 \frac{1 + \mu(\tilde{\g A}) }{n} \sum_{j=1}^n  \|\g W_j^R \Re( \g \delta) + \i \g W_j^I \Im(\g \delta)\|^2  - \mu(\tilde{\g A}) M \left(  \sum_{j=1}^n  \|\g W_j^R \Re( \g \delta) + \i \g W_j^I \Im(\g \delta)\|\right)^2  \leq 4\varepsilon^2 .
\end{equation}
We will now use this inequality to derive an upper bound on $\sum_{j=1}^n  \|\g W_j^R \Re( \g \delta) + \i \g W_j^I \Im(\g \delta)\|^2$, which will also apply to $\|\tilde{\g W}\tilde{\g \delta}\|_2^2 \leq \sum_{j=1}^n  \|\g W_j^R \Re( \g \delta) + \i \g W_j^I \Im(\g \delta)\|^2 = \sum_{j=1}^n (\|\g W_j^R \Re(\g \delta) \|_2^2 + \|\g W_j^I \Im(\g \delta) \|_2^2 )$, since the groups overlap and the squared components of $\tilde{\g W}\tilde{\g \delta}$ are summed multiple times in the right-hand side. To derive the upper bound, we first introduce a few notations: 
$$
	a =  \sum_{j\in I_0}  \|\g W_j^R \Re( \g \delta) + \i \g W_j^I \Im(\g \delta)\|,\quad b =   \sum_{j\notin I_0} \|\g W_j^R \Re( \g \delta) + \i \g W_j^I \Im(\g \delta)\|,
$$
and 
$$
	c_0 = \frac{\sum_{j\in I_0}  \|\g W_j^R \Re( \g \delta) + \i \g W_j^I \Im(\g \delta)\|^2}{\left(  \sum_{j\in I_0}  \|\g W_j^R \Re( \g \delta) + \i \g W_j^I \Im(\g \delta)\|\right)^2} \in \left[\frac{1}{|I_0|}, 1\right],
$$

$$
	c_1 =  \frac{\sum_{j\notin I_0}  \|\g W_j^R \Re( \g \delta) + \i \g W_j^I \Im(\g \delta)\|^2}{\left(  \sum_{j\notin I_0}  \|\g W_j^R \Re( \g \delta) + \i \g W_j^I \Im(\g \delta)\|\right)^2} \in \left[\frac{1}{n - |I_0|}, 1\right] ,
$$
where the box bounds are obtained by classical relations between the $\ell_1$ and $\ell_2$ norms\footnote{Let $\g u\in\R^{|I_0|}$ with $u_j =  \|\g W_j^R \Re( \g \delta) + \i \g W_j^I \Im(\g \delta)\|$. Then, $c_0 = (\|\g u\|_2 / \|\g u\|_1)^2$ and the bounds are obtained by the classical relation $\forall\g u\in \R^k$, $\|\g u\|_2 \leq \|\g u\|_1 \leq \sqrt{k}\|\g u\|_2$.}.  
With these notations, the term to bound is rewritten as
$$
	\sum_{j=1}^n  \|\g W_j^R \Re( \g \delta) + \i \g W_j^I \Im(\g \delta)\|^2 =  c_0 a^2 + c_1 b^2 ,
$$
while the inequality~\eqref{eq:proofstabilitygroup2} becomes
$$
	\frac{1 + \mu(\tilde{\g A}) }{n} (c_0 a^2 + c_1 b^2)  - \mu(\tilde{\g A}) M (a + b)^2  \leq 4\varepsilon^2 .
$$
We further reformulate this constraint by letting $a=\rho b$:
\begin{equation}\label{eq:proofstabilitygroup3}
	\frac{1 + \mu(\tilde{\g A}) }{n} (c_0 \rho^2 + c_1) b^2  - \mu(\tilde{\g A}) M (1 + \rho)^2 b^2  \leq 4\varepsilon^2 .
\end{equation}
Let $\gamma =  (1 + \rho)^2 / (c_0 \rho^2 + c_1)$. Due to~\eqref{eq:proofainfb}, we have $a<b$ and thus $\rho\in[0,1)$, which, together with the bounds on $c_0$ and $c_1$, gives the constraints $1\leq \gamma\leq 4(n-|I_0|)$. By setting $V = (c_0 \rho^2 + c_1) b^2$, \eqref{eq:proofstabilitygroup3} is rewritten as 
$$
	\frac{1 + \mu(\tilde{\g A}) }{n} V  - \mu(\tilde{\g A}) M \gamma V \leq 4\varepsilon^2 , 
$$
where 
$$
	\frac{1 + \mu(\tilde{\g A}) }{n}  - \mu(\tilde{\g A}) M \gamma \geq \frac{1 + \mu(\tilde{\g A}) }{n}  - 4 (n-|I_0|) \mu(\tilde{\g A}) M > 0, 
$$
since $\gamma\leq 4(n-|I_0|)$ and the positivity is ensured by the condition \eqref{eq:conditionP1groupnoise} and the fact that $\|\g x_0\|_0 = n - |I_0|$. Thus, 
$$
	\|\tilde{\g W}\tilde{\g \delta}\|_2^2 \leq \sum_{j=1}^n  \|\g W_j^R \Re( \g \delta) + \i \g W_j^I \Im(\g \delta)\|^2 = V \leq \frac{4n\varepsilon^2}{1+\mu(\tilde{\g A}) - 4 \mu(\tilde{\g A}) n M \|\g x_0\|_0 } ,
$$
which proves the stability result in~\eqref{eq:stabilityxgroupW} since $\|\tilde{\g W} \tilde{ \g \delta} \|_2^2 = \|\g W^R \Re(\g \delta) + \i \g W^I \Im(\g \delta) \|^2$.

To conclude in the case $\varepsilon=0$, it remains to see that~\eqref{eq:stabilityxgroupW} implies $\hat{\g v} = \nu(\g x_0)$ and that Definition~\ref{def:nuinv} ensures $\nu^{-1}(\nu(\g x_0)) \in T(\g x_0)$. 
\end{proof}

\section{Complex data, but real solutions}
\label{sec:complexreal}

Consider now the following problem:
\begin{align}\label{eq:P0complexreal}
	\min_{\g x\in\R^n} \ & \|\g x\|_0 \\
	\mbox{s.t.}\ & y_i =  |\g q_i^H \g x|^2,\quad i=1,\dots,N, \nonumber
\end{align}
where  the signal $\g x\in\R^n$ and measurements $y_i\in \R$ are assumed to be real while the vectors $\g q_i \in \C^n$ can be complex.
In this case, $\g v = \nu(\g x)$ is also a real vector and solutions can be approximated via a dedicated version of~\eqref{eq:P1vcomplex}:
\begin{align}\label{eq:P1vcomplexreal}
	\min_{\g v\in\R^M} \ & \sum_{j=1}^n \|\g W_j^R \g v\|_2\\
	\mbox{s.t.}\ & \g y = \Re\left(\g A\g v\right) = \Re\left(\g A\right)\g v \nonumber\\
	& v_{jj} \geq 0,\ j=1,\dots,n ,\nonumber
\end{align}
where $\g A \in\C^{N\times M }$ is defined as in Sect.~\ref{sec:firstrelax} and $\g W_j^R = \g W_j \g W^R$ with precompensating weights given by $( \g W^R)_{i,i}= \|\Re(\g A_i)\|_2$. 

For this particular case, Theorem~\ref{thm:groupsparsecomplexreal} below is similar in spirit to Theorem~\ref{thm:groupsparsecomplex}, but allows us to gain a $\sqrt{2}$ factor by using Lemma~\ref{lem:deltacomplexreal} instead of Lemma~\ref{lem:deltacomplex} in order to take into account that $\g \delta$ belongs to $\R^M$ (detailed proof given in Appendix~\ref{sec:proofthmcomplexreal}). This results in a bound on $\|\g x_0\|_0$ similar to the one in Theorem~\ref{thm:groupsparse} for the real case and based on the mutual coherence of the real part of $\g A$. Since $\mu(\Re(\g A)) \leq \mu(\tilde{\g A})$, this also improves (relaxes) the bound compared with the one of Theorem~\ref{thm:groupsparsecomplex}.

\begin{theorem}\label{thm:groupsparsecomplexreal}
Let $\g x_0\in\R^n$ be such that $y_i = \g x_0^T\g q_i \g q_i^H\g x_0$, $i=1,\dots,N$, and $\g v_0=\nu(\g x_0)$. If the condition
$$
	 \|\g x_0\|_0 <  \frac{1}{2\sqrt{n}}\sqrt{1+ \frac{1}{\mu^2(\Re(\g A))}} 
$$
holds, then $\g v_0$ is the unique solution to \eqref{eq:P1vcomplexreal}.
\end{theorem}

In a typical instance of Problem~\eqref{eq:P0complexreal}, the measurements $\g y$ correspond to the power spectrum of a real signal. However, in this case, Theorem~\ref{thm:groupsparsecomplexreal} does not apply since the solution is known not to be unique due to invariances of the measurements. The following subsections first describe a practical technique to deal with such invariances and then focus on the power spectrum case. 

\subsection{Invariances }

Consider the case where the measurements are invariant to some transformation of the signal. A typical example is when $\g x\in\R^n$ and $\{\g q_i\in\C^n \}$ forms a Fourier basis. Then, the equations $y_i=|\g q_i^H \g x|^2$ are invariant to circular shifts of the components of $\g x$. This is problematic since shifted versions of $\g x$ lead to shuffled\footnote{Circular shifts of $\g x$ lead to rearrangements of the components of $\nu(\g x)$ which are not exactly circular shifts. For example, with $\g x_0=[1,1,0,0]^T$ and $\g x_1=[0, 1,1,0]^T$, we have $\nu(\g x_0)=[1,1,0,0,1,0,0,0,0,0]^T$ and $\nu(\g x_1)=[0,0,0,0,1,1,0,1,0,0]^T$.} versions of $\nu(\g x)$. Thus, multiple shuffled $\g v$'s satisfy the linearized constraints, $y_i = \Re(\g A)\g v$, and so does any convex combination of them. These convex combinations need not be sparse as they combine vectors with different sparsity patterns, but lead to lower or equal values of the convex cost function of~\eqref{eq:P1vcomplexreal}. 

To circumvent this issue, we need to linearize the constraints by a shift-invariant transformation $\phi$, such that $\phi(\g x_0) = \phi(\g x_1)$ for $\g x_0$ and $\g x_1$ two shifted versions of the same vector. To lead to an effective estimation method, the transformation $\phi$ must also retain sparsity in a sense similar to $\nu$.

Consider the transformation $\phi : \R^n \rightarrow \R^M$ defined by\footnote{In the case where the argmax in~\eqref{eq:phishift}	is not a singleton, $k$ is arbitrarily set to the minimum of the indexes in the argmax and $\phi$ cannot be shift-invariant. Assumptions regarding this issue will be made clear in Proposition~\ref{prop:varphi} and Theorem~\ref{thm:uniquereflect} below.}
\begin{equation}\label{eq:phishift}	
		\phi(\g x) = \nu( \mbox{shift}(\g x, 1-k)) , 
		\quad \mbox{with } 
		k = \arg\max_{i\in\{1,\dots, n\}} |x_i| ,  
\end{equation}
where shift$(\cdot,k)$ stands for the $k$-steps circular shift operator. 
The transformation $\phi$ in~\eqref{eq:phishift} first shifts the vector $\g x$ such that the first component is the one with maximal magnitude. This results in a shift-invariant transformation which retains sparsity and the linearization ability via the mapping $\nu$. The linearized constraints remain the same, i.e., $y_i = \Re(\g A)\g \phi$, but another constraint must be added to \eqref{eq:P1vcomplexreal} in order to account for the shift-invariance. 

More precisely, the definition of $\phi$ ensures that 
$$
	\left(\phi(\g x)\right)_{11} \geq \left(\phi(\g x)\right)_{jj} ,\quad  j = 2,\dots,n .
$$
Thus, a valid vector $\hat{\g \phi}$ estimating $\phi(\g x_0)$ can be found by solving
\begin{align}\label{eq:P1phicomplexreal}
	\hat{\g \phi}= \arg\min_{\g \phi\in\R^M} \ & \sum_{j=2}^n \|\g W_j^R \g \phi\|_2\\
	\mbox{s.t.}\ & \g y = \Re\left(\g A\right)\g \phi \nonumber\\
	& \phi_{11} \geq \phi_{jj}  \geq 0,\ j=2,\dots,n .\nonumber
\end{align}
Note that the the cost function does not involve the first group of variables since $x_1$ is assumed to be nonzero. 
In comparison with  \eqref{eq:P1vcomplexreal}, this formulation is still convex but cannot have multiple solutions that are shifted/shuffled versions of one another. 

Finally, the set of shifted solutions $\{\g x_k\}_{k=1}^n$ to \eqref{eq:P0complexreal} is approximated by the set 
$$
	\phi^{-1}(\hat{\g \phi}) = \left\{ \mbox{shift}(\nu^{-1}(\hat{\g \phi}), k) \right\}_{k=1}^n .
$$

Following a similar approach, we can additionally take into account reflections by defining 
\begin{equation}\label{eq:phireflect}
	\begin{cases}
		\phi(\g x) = \nu(\g x_2 ) \\
		\g x_2 = \varphi(\g x) = \begin{cases}\g x_1, \ \mbox{if }\sum_{i=2}^{n/2} |x_{1i}|^2 \geq \sum_{i=2+\frac{n}{2}}^{n} |x_{1i}|^2 \\
									\mbox{shift}(\mbox{reflection}(\g x_1), 1),\ \mbox{otherwise}\end{cases}\\
		 \g x_1 =  \mbox{shift}(\g x, 1-k) \\
			k = \arg\max_{i\in\{1,\dots, n\}} |x_i| , 
	\end{cases}	
\end{equation}
where we assume $n$ to be even and reflection$(\cdot)$ is the reflection operator defined by reflection$(\g x) = [x_{n}, x_{n-1},\dots,x_2,x_1]^T$. 
In plain words, the transformation $\phi$ in~\eqref{eq:phireflect} first shifts the vector $\g x$ such that the first component is the one with maximal magnitude. Then, it applies a centered reflection to the shifted $\g x$, named $\g x_1$, only if the sum of squares over the first entries of $\g x_1$ (without the first one) is smaller than the one over the last ones. If this is the case, the result is shifted again to recover the first component of $\g x_2$ with maximal magnitude. 
Finally, the Veronese map is applied to $\g x_2$ to give $\phi(\g x)$. 

Invariance of $\phi = \nu\circ \varphi$ to circular shifts and reflections is implied by the invariance of $\varphi$ given in the Proposition below (proof in Appendix~\ref{sec:proofprop}). 

\begin{proposition}\label{prop:varphi}
For all $\g x\in\C^n$ such that $\left|\arg\max_{i\in\{1,\dots,n\}} |x_i| \right| = 1$, the following statements hold for the transformation $\varphi$ defined in~\eqref{eq:phireflect}:
\begin{enumerate}
	\item $\varphi$ is idempotent, i.e., $\varphi\circ\varphi(\g x) = \varphi(\g x)$;
\end{enumerate} 
and, for all $\g x$ additionally satisfying $\sum_{i=2}^{n/2} |x_{1i}|^2 \neq \sum_{i=2+\frac{n}{2}}^{n} |x_{1i}|^2 $ with $\g x_1 = \mbox{shift}(\g x, \arg\max_{i\in\{1,\dots,n\}} |x_i| )$, 
\begin{enumerate}
	\item[2.] $\varphi$ is shift-invariant, i.e., $\forall s\in \mathbb{Z},\ \varphi(\g x) = \varphi(\mbox{shift}(\g x, s))$;
	\item[3.] $\varphi$ is reflection-invariant, i.e., $\varphi(\g x) = \varphi(\mbox{reflection}(\g x))$.
\end{enumerate}
Note that statements 2 and 3 imply that $\varphi$ is invariant to any combination of shifts and reflections. 
\end{proposition}

Proposition~\ref{prop:varphi} also shows the idempotence of $\varphi$, which of course does not transfer to $\phi$ directly but which is however very useful. Indeed, this allows us to test if a candidate $\g x_2$,  supposed to be the result of $\varphi$ applied to some vector, is consistent with the definition of $\varphi$ as $\g x_2 \overset{?}{=} \varphi(\g x_2)$, which can be checked via simple inequalities.
Since these inequalities only involve the (squared) magnitude of the entries in the vector, they can also be easily embedded as linear constraints in \eqref{eq:P1phicomplexreal} to compute an estimate $\hat{\g\phi}$ that is consistent with the transformation \eqref{eq:phireflect}. This yields the convex program
 \begin{align}\label{eq:P1phicomplexreal2}
	\hat{\g \phi}= \arg\min_{\g \phi\in\R^M} \ & \sum_{j=2}^n \|\g W_j^R \g \phi\|_2\\
	\mbox{s.t.}\ & \g y = \Re\left(\g A\right)\g \phi & \mbox{(data fitting)}\nonumber\\
	& \phi_{11} \geq \phi_{jj}  \geq 0,\ j=2,\dots,n & \mbox{(shift-invariance)}\nonumber\\
	& \sum_{i=2}^{n/2} \phi_{ii} \geq \sum_{i=2+\frac{n}{2}}^{n} \phi_{ii} & \mbox{(reflection-invariance)}\nonumber ,
\end{align}
where $\phi_{jj}$ estimates $|x_j|^2$. 
The advantage of using \eqref{eq:P1phicomplexreal2} is that it has a single solution independently of the number of shifted/reflected solutions to \eqref{eq:P0complexreal}, in the sense of the next theorem. 

\begin{theorem}\label{thm:uniquereflect}
If Problem~\eqref{eq:P0complexreal} has a unique set of shifted/reflected solutions $S(\g x_0) = \{ \mbox{shift}(\g x_0, k) : k\in \{1,\dots,n\} \} \cup \{ \mbox{shift}(\mbox{reflection}(\g x_0), k) : k\in \{1,\dots,n\} \}$, and if $\left|\arg\max_{i\in\{1,\dots,n\}} |x_{0i}| \right| = 1$, then 
there is exactly one vector $\nu(\g x)$ with $\g x\in S(\g x_0)$ in the feasible set of~\eqref{eq:P1phicomplexreal2}.  
\end{theorem}
\begin{proof}
Since all $\g x\in S(\g x_0)$ are solutions to \eqref{eq:P0complexreal}, they satisfy the first constraint in~\eqref{eq:P1phicomplexreal2} as $\g y = \Re\left(\g A\right) \nu(\g x)$. 
By Proposition~\ref{prop:varphi}, $\varphi$ is constant over the set $S(\g x_0)$ and equal to one of the vectors from this set, say $\g x_0$, for which we have in particular $\g x_0 = \varphi(\g x_0)$. Then, $\nu(\g x_0) = \nu\circ\varphi(\g x_0) = \phi(\g x_0)$ and thus $\nu(\g x_0)$ satisfies all the constraints of~\eqref{eq:P1phicomplexreal2} and is a feasible point.  

Now take an $\g x\in S(\g x_0)$ with $\g x\neq \g x_0$. Then, either $\g x = \mbox{shift}(\g x_0, k)$ or $\g x = \mbox{shift}(\mbox{reflection}(\g x_0), k)$ for some $k\in\{1,\dots,n\}$. This implies  $\arg\max_i (\nu(\g x))_{ii} = \arg\max_i |x_i| \neq \arg\max_i |x_{0i}| = \arg\max_i (\nu(\g x_0))_{ii}$. Since  $\g x_0$ satisfies the constraints in~\eqref{eq:P1phicomplexreal2}, we have $\arg\max_i (\nu(\g x_0))_{ii}=1$ and thus $\arg\max_i (\nu(\g x))_{ii} \neq 1$, which shows that $\nu(\g x)$ violates the second constraint and is not a feasible point of~\eqref{eq:P1phicomplexreal2}. 
\end{proof}

\paragraph{Illustrative example.} Consider the vectors ($\g x_1$ and $\g x_2$ are not related to the notations of \eqref{eq:phireflect})
$$
	\g x_1 = [1, 2 , 3 , 4, 0, 0]^T,\quad	
		\g x_2 = [4, 0, 0, 1, 2 , 3 ]^T,\quad
			\g x_3 = [0, 0, 4, 3, 2, 1]^T,
$$
$$
				\g x_4 = [2 ,1, 0, 0, 4, 3]^T,\quad
				\g x_5 = [4 ,3, 2, 1, 0, 0]^T, 
$$
which are all obtained by shifts and reflections of the same vector. They all lead to the same $\phi(\g x_i) =  \nu (\g x_5)$. Indeed, the transformation $\phi$ first applies the required shift and reflect operations to map $\g x_i$ to $\g x_5$ which has a first component with the largest magnitude and the largest half-sum over its first entries. Then, $\phi$ computes the Veronese map of $\g x_5$. Also note that all the $\nu(\g x_i)$ are feasible with respect to $\g y = \Re\left(\g A\right)\nu(\g x_i)$, but only $\nu(\g x_5)$ satisfies the additional constraints implementing shift and reflection invariance in \eqref{eq:P1phicomplexreal2}. 
Therefore, we can obtain $\hat{\g\phi}= \nu(\g x_5)$ as the unique solution to \eqref{eq:P1phicomplexreal2}.

\subsection{Support recovery from the power spectrum}
\label{sec:fourier}

When the measurements correspond to the squared magnitude of the Fourier transform of the signal, i.e., its (squared) power spectrum, there is another issue beside shift/reflection-invariances. Even with the correct support $\mathcal{S} = supp(\g x_0)$, the linear system $\g y = \Re(\g A_\mathcal{S})\g \phi_\mathcal{S}$, where $\g A_\mathcal{S}$ is the submatrix of $\g A$ with the corresponding columns, is under-determined. 
More precisely, for all $i$ and all $j$, $|q_{ij}| = 1$, which implies $(\g a_i)_{jj} = (\nu(\g q_i))_{jj} = 1$ and $\g A_{jj} = \g 1$. Thus, even when limited to the support $\mathcal{S}$, $\g A_\mathcal{S}$ has $|\mathcal{S}|$ similar columns and is rank-defficient. 

Therefore, in this case the proposed approach cannot exactly recover $\g x_0$ and \eqref{eq:P1phicomplexreal2} is used to estimate the support $\mathcal{S}$. Indeed, knowing the correct support can be useful for other methods dedicated to the classical phase retrieval problem \cite{Fienup82}. 

Though this case seems unfavorable, if the number of measurements satisfies $N \geq 2n-1$, the autocorreleation of $\g x_0$ (with zero padding), 
$$
	r_k = \sum_{i=1}^n x_i x_{i+k},\quad k=-n+1,\dots,n-1 ,
$$
can be computed via the inverse Fourier transform of $\g y$ and thus is available. This information can be included as linear constraints on the components $\phi_{i(i+k)}$ of $\g \phi$ estimating the cross products $x_i x_{i+k}$ to drive the solution towards a satisfactory one. 

In addition, following \cite{Jaganathan12,Shechtman13}, this can be used to restrict the support of the solution as follows.
First, note that with $N\geq 2n-1$, the measurements are not invariant to circular shifts of $\g x$ but of $\g x$ with zero-padding. Thus, we cannot assume $|x_1| \geq |x_j|$, $j=2,\dots,n$. However, we can fix $x_1 \neq 0$. Then, for all base variable index $j\in\{1,\dots,n\}$,
\begin{equation}\label{eq:acorrsupp}	
	 r_{k} = 0,\ k= j-1,\dots,n-1\quad \Rightarrow\quad x_{j} = 0.
\end{equation}
Thus, the corresponding variables in $\g \phi$ can be set to zero, i.e., $\g W_j \g \phi = \g 0$. 
Moreover, let $j_m$ be the minimum of the $j$'s satisfying the condition in~\eqref{eq:acorrsupp}, then $ x_{j_m -1} \neq 0$, which can be favored by removing the corresponding term in the cost function. Combining all the information on the solution in a convex program leads to the final formulation 
\begin{align}\label{eq:P1fourier}
	\hat{\g \phi}= \arg\min_{\g \phi\in\R^M} \ & \sum_{j=2}^{j_m-2} \|\g W_j^R \g \phi\|_2  & \\
	\mbox{s.t.}\ & \g y = \Re\left(\g A\right)\g \phi  & \mbox{ (data fitting)}\nonumber\\
	& \phi_{jj}  \geq 0,\ j=1,\dots,j_m-1   & \mbox{ (structural knowledge)}\nonumber\\
	& \sum_{i=2}^{n/2} \phi_{ii} \geq \sum_{i=2+\frac{n}{2}}^{n} \phi_{ii}  & \mbox{ (reflection-invariance)} \nonumber \\
	&  \sum_{i=1}^{n-k} \phi_{i(i+k)} = r_k, \ k=0,\dots,n-1   & \mbox{ (autocorrelation)}\nonumber\\
	& \g W_j \g \phi = \g 0,\ j=j_m,\dots, n    & \mbox{ (restricted support)} .  \nonumber
\end{align}

A noise-tolerant version of~\eqref{eq:P1fourier} is obtained by replacing the constraint $\g y = \Re\left(\g A\right)\g \phi $ by $\|\g y - \Re\left(\g A\right)\g \phi\|_2 \leq \varepsilon$. This modification also applies to~\eqref{eq:P1vcomplexreal} and~\eqref{eq:P1phicomplexreal2}.

\section{Experiments} 
\label{sec:exp}

For the experiments, we consider two variants of the proposed approach: the convex method described in details in the previous sections and the greedy method for solving the group-sparse optimization problems~\eqref{eq:P0groupv} and~\eqref{eq:P0vcomplex}. Implementation details for these two methods are as described in~\cite{Lauer13b}, with slight modifications to handle complex variables. In particular, the convex method uses the iterative reweighting of \cite{Candes08} adapted to the group-sparse setting to enhance the sparsity of the solution. The greedy method starts with an empty support and, at each iteration, adds the group of variables in $\g v$ corresponding to the base variable $x_j$ that results in the best approximation of $\g y$. 

For complex signals, $\g x_0$, we measure the relative error corresponding to the normalized distance between the estimate $\hat{\g x}$ and the set $T(\g x_0)$, i.e., $\min_{\g x\in T(\g x_0)} \ \|\hat{\g x} - \g x\| / \|\g x\|$. 
Exact recovery is detected when this error is smaller than $10^{-6}\|\g x_0\|$.  

\paragraph{Exact recovery in the noiseless case.}
We estimate the probability of exact recovery at various sparsity levels, $\|\g x_0\|_0$, in the complex case where $n=20$ and $N=50$. For each sparsity level, the probability is estimated as the percentage of successful trials over a Monte Carlo experiment with 100 trials. In each trial, $N$ complex vectors $\g q_i \in\C^n$ are drawn from a zero-mean Gaussian distribution of unit variance to measure a random signal $\g x_0 \in\C^n$ with $\|\g x_0\|_0$ nonzero entries at random locations whose values are drawn from a zero-mean Gaussian distribution of unit variance. Results shown in Fig.~\ref{fig:exact} indicate that the proposed approach can exactly recover sufficiently sparse signals with high probability. 

Similar experiments are performed to evaluate the influence of the number of measurements $N$ on the probability of exact recovery. Results reported in the right plot of Fig.~\ref{fig:exact} show that, with a sparsity level of $\|\g x_0\|_0= 4$, the convex method requires more measurements for perfect recovery than the greedy strategy. However, $N \approx \|\g x_0\|_0^3$ measurements are already sufficient to exactly recover $\g x_0$ in all trials.  
\begin{figure}
	\centering
	\includegraphics[width=.4\linewidth]{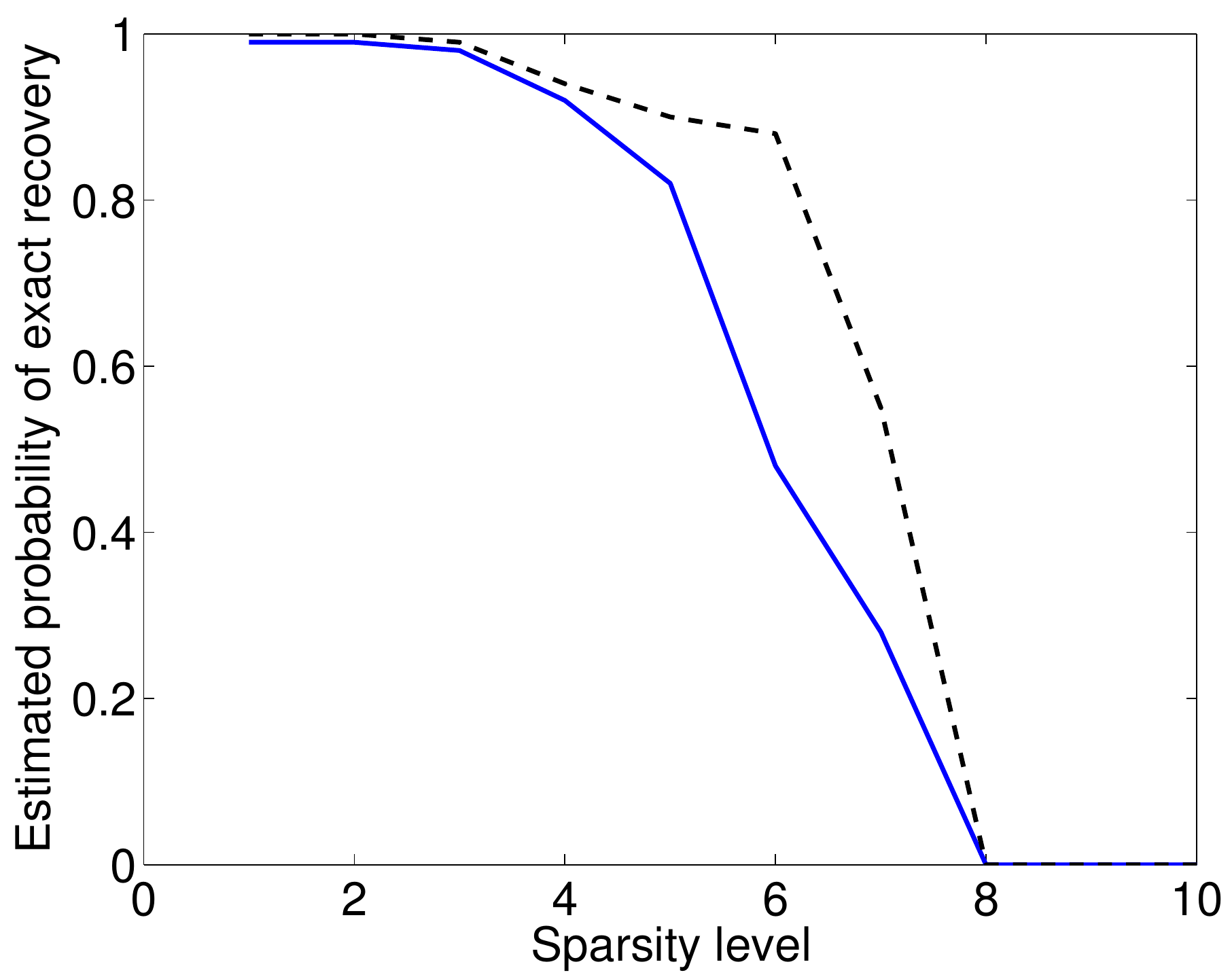}
	\includegraphics[width=.4\linewidth]{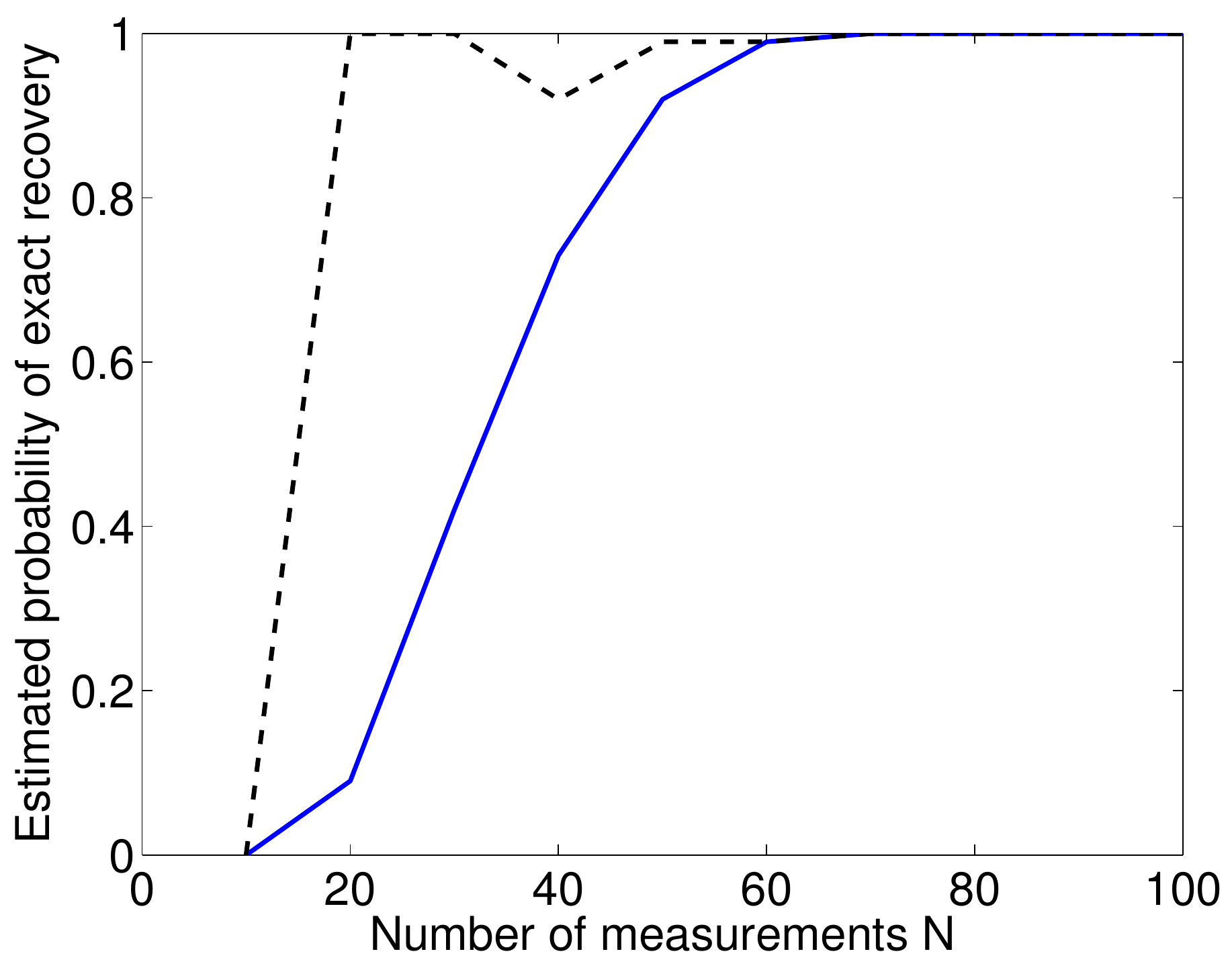}	
	\caption{Estimated probability of exact recovery in the noiseless case for the convex relaxation~\eqref{eq:P1vcomplex} (plain line) and the greedy method applied to~\eqref{eq:P0vcomplex} (dashed line) versus the number of nonzeros (left) and the number of measurements (right). \label{fig:exact}}
\end{figure}

\paragraph{Stable recovery in the noisy case.}
We now consider the noisy case where $y_i =  |\g q_i^H \g x_0|^2 + e_i$, $i=1,\dots,N$, and $\|\g e\|_2 \leq \varepsilon$. Over 100 trials, Figure~\ref{fig:noise} reports the mean relative error and the rate of successful support recovery for $N=50$ and $\varepsilon = 3$. These results show that sufficiently sparse signals can be accurately estimated from noisy quadratic measurements and that the probability of support recovery in this case follows a curve similar to the one of the probability of exact recovery in the noiseless case. This means that the methods are robust to noise regarding the recovery of the correct support. 

\begin{figure}
	\centering
	\includegraphics[width=.4\linewidth]{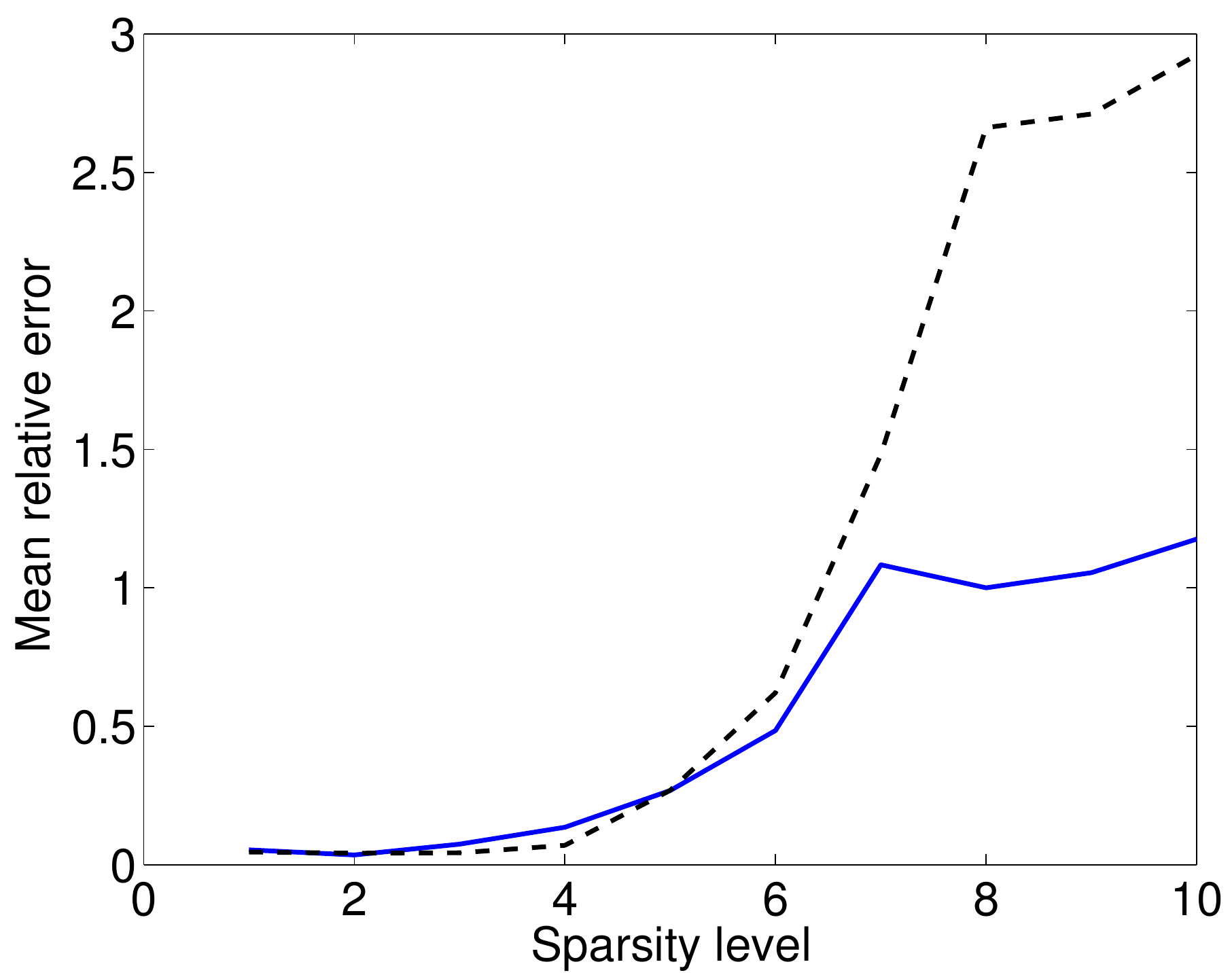}
	\includegraphics[width=.4\linewidth]{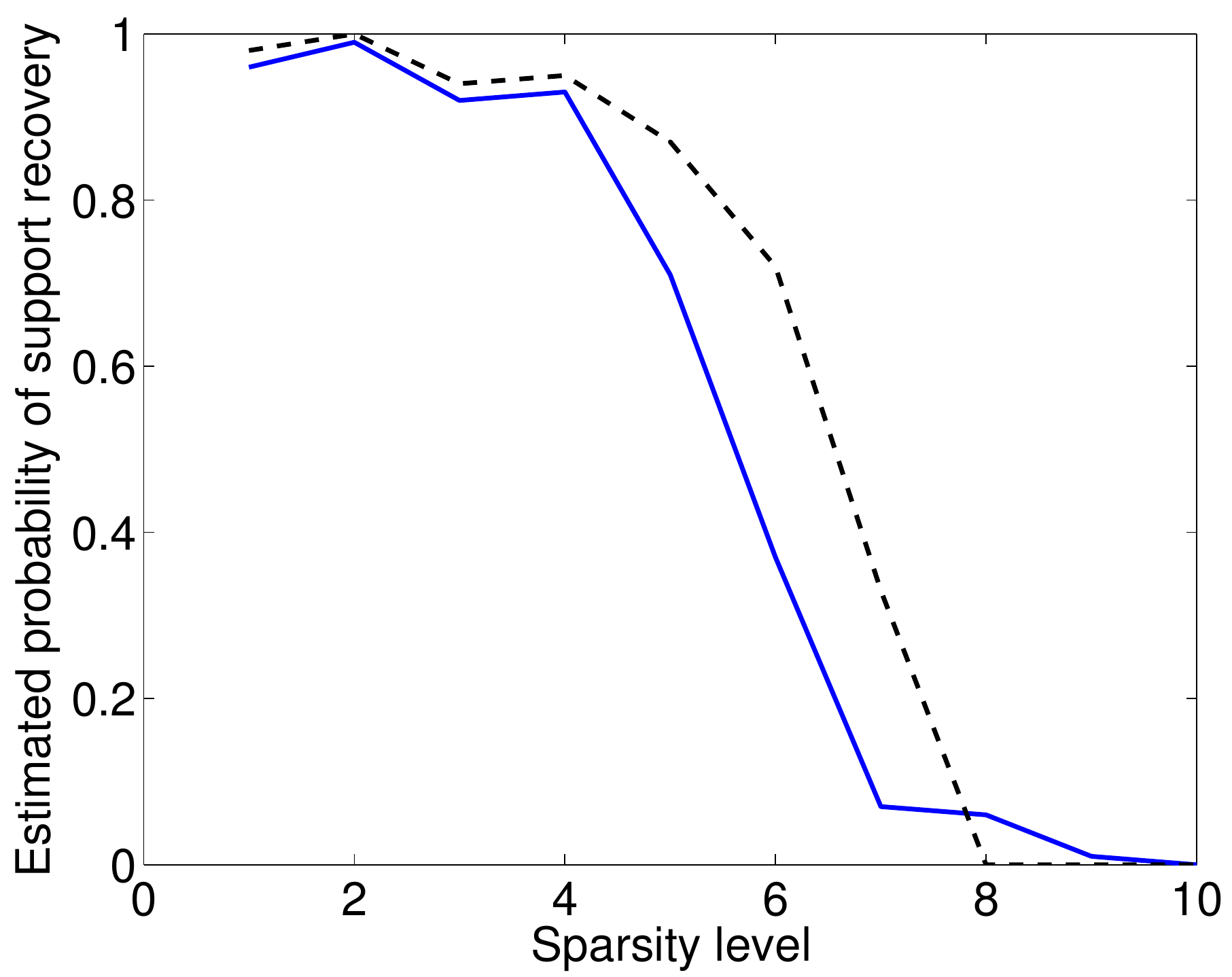}	
	\caption{Mean relative error (left) and estimated probability of support recovery (right) in the noisy case for the convex relaxation~\eqref{eq:P1vcomplex} (plain line) and the greedy method applied to~\eqref{eq:P0vcomplex} (dashed line).  \label{fig:noise}}
\end{figure}

\paragraph{Support estimation from the power spectrum.}
We now test if the convex method is robust to invariances, such as the ones discussed in Sect.~\ref{sec:fourier}, when estimating the support. In particular, we start with a setting in which $\g x\in\R^n$, $N=n=20$ and $\g q_i^H$, $i=1,\dots,N$, are the rows of the $n$-point Fourier matrix. In this case, the autocorrelation cannot be computed and the convex formulation~\eqref{eq:P1phicomplexreal2} is used to estimate the support. Then, we perform similar experiments but with oversampling ($N=2n$), thus allowing for the computation of the autocorrelation and the use of~\eqref{eq:P1fourier}. 
Results shown in Fig.~\ref{fig:fourier} indicate that by using~\eqref{eq:P1phicomplexreal2} we can recover the correct support for sufficiently sparse signals without oversampling, i.e., with as many measurements as unknowns, while using more information extracted from the autocorrelation of the signal only slightly helps to recover larger supports. 

\begin{figure}
	\centering
	\includegraphics[width=.4\linewidth]{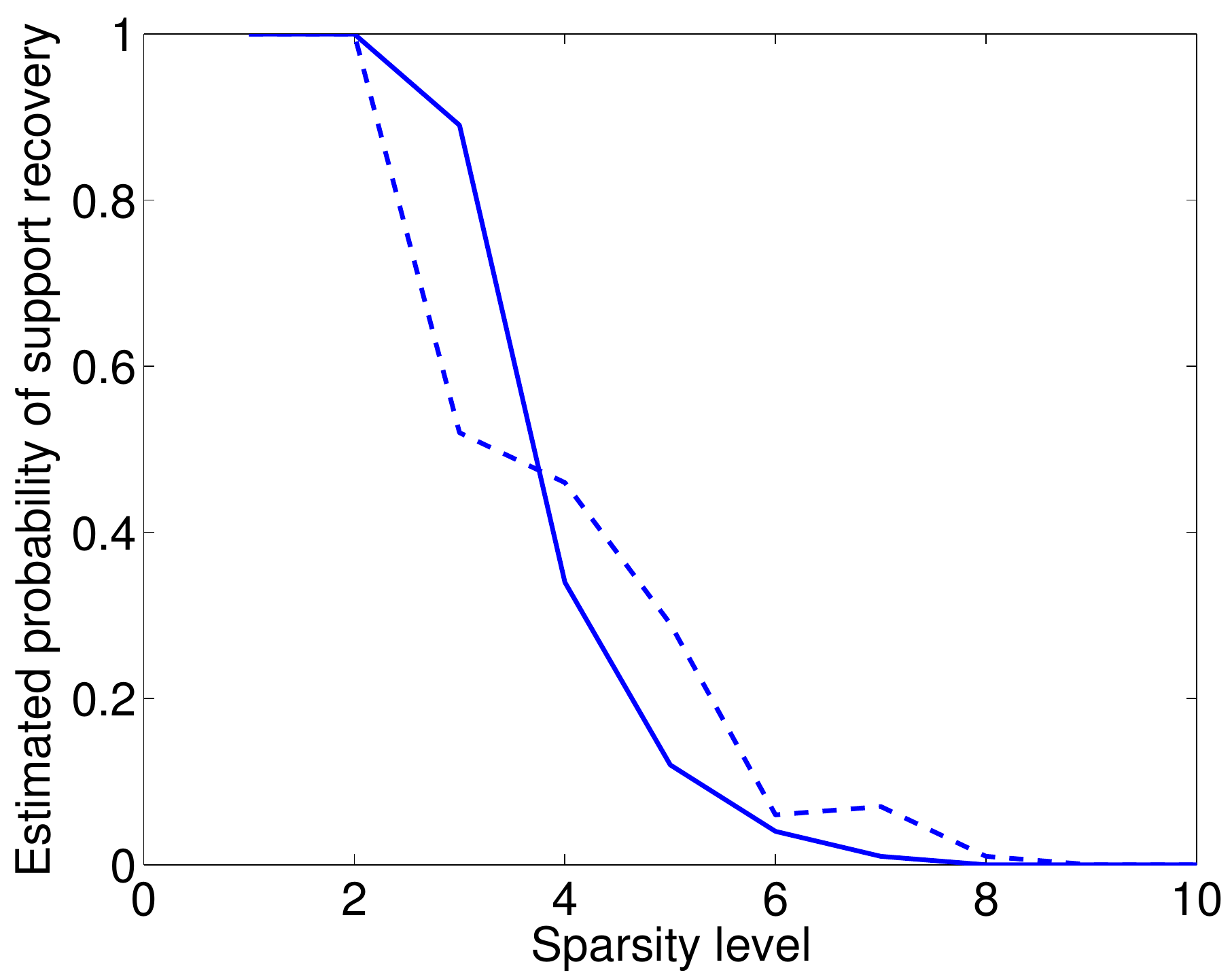}
	\caption{Estimated probability of support recovery from the magnitude of the Fourier transform when using~\eqref{eq:P1phicomplexreal2} for $N=n$ (plain line) and when using~\eqref{eq:P1fourier} for $N=2n$ (dashed line) versus the sparsity level. \label{fig:fourier}}
\end{figure}

\section{Conclusions}

The paper proposed a new approach to phase retrieval of sparse signals. This approach is based on a group-sparse optimization formulation of the problem with linearized constraints. 
Exact and stable recovery results were shown for a convex relaxation of this formulation both in the real and complex case. 
Invariances to circular shifts and reflections that are common in phase retrieval problems were also discussed and a practical technique was given to prevent these from breaking the sparsity of the solution. 

Future work will focus on deriving theoretical guarantees for the invariance issue. Another direction of research concerns the analysis of greedy algorithms for the group-sparse optimization problem, which proved as valuable as the convex relaxations in experiments for measurements without invariance. How to deal with invariances in these methods will also be investigated.

\bibliographystyle{plain}

\begin{thebibliography}{}

\end{thebibliography}


\begin{thebibliography}{10}

\bibitem{Candes06}
E.~J Cand{\`e}s.
\newblock Compressive sampling.
\newblock In {\em Proceedings oh the International Congress of Mathematicians:
  invited lectures}, pages 1433--1452, 2006.

\bibitem{Candes13b}
E.~J. Cand{\`e}s, Y.~C. Eldar, T.~Strohmer, and V.~Voroninski.
\newblock Phase retrieval via matrix completion.
\newblock {\em SIAM Journal on Imaging Sciences}, 6(1):199--225, 2013.

\bibitem{Candes13}
E.~J. Cand{\`e}s, T.~Strohmer, and V.~Voroninski.
\newblock Phaselift: Exact and stable signal recovery from magnitude
  measurements via convex programming.
\newblock {\em Communications on Pure and Applied Mathematics},
  66(8):1241--1274, 2013.

\bibitem{Candes08}
E.~J. Cand\`es, M.~B. Wakin, and S.~P. Boyd.
\newblock {Enhancing sparsity by reweighted $\ell_1$ minimization}.
\newblock {\em Journal of Fourier Analysis and Applications}, 14(5):877--905,
  2008.

\bibitem{Donoho06}
D.~L. Donoho.
\newblock Compressed sensing.
\newblock {\em IEEE Transactions on Information Theory}, 52(4):1289--1306,
  2006.

\bibitem{Donoho06noise}
D.~L. Donoho, M.~Elad, and V.~N. Temlyakov.
\newblock Stable recovery of sparse overcomplete representations in the
  presence of noise.
\newblock {\em IEEE Transactions on Information Theory}, 52(1):6--18, 2006.

\bibitem{Ehler13}
M.~Ehler, M.~Fornasier, and J.~Sigl.
\newblock Quasi-linear compressed sensing.
\newblock Technical report, 2013.
\newblock
  http://www-m15.ma.tum.de/foswiki/pub/M15/Allgemeines/PublicationsEN/greedy\_21.pdf.

\bibitem{Fienup82}
J.~Fienup.
\newblock Phase retrieval algorithms: a comparison.
\newblock {\em Applied Optics}, 21(15):2758--2769, 1982.

\bibitem{Gerchberg72}
R.~Gerchberg and W.~Saxton.
\newblock A practical algorithm for the determination of phase from image and
  diffraction plane pictures.
\newblock {\em Optik}, 35:237--246, 1972.

\bibitem{Gonsalves76}
R.~Gonsalves.
\newblock Phase retrieval from modulus data.
\newblock {\em Journal of Optical Society of America}, 66(9):961--964, 1976.

\bibitem{Harrison93}
R.~W. Harrison.
\newblock Phase problem in crystallography.
\newblock {\em Journal of the Optical Society of America A}, 10(5):1046--1055,
  1993.

\bibitem{Jaganathan12}
K.~Jaganathan, S.~Oymak, and B.~Hassibi.
\newblock Recovery of sparse 1-d signals from the magnitudes of their {Fourier}
  transform.
\newblock In {\em IEEE International Symposium on Information Theory (ISIT)},
  pages 1473--1477, 2012.

\bibitem{Kohler72}
D.~Kohler and L.~Mandel.
\newblock Source reconstruction from the modulus of the correlation function: a
  practical approach to the phase problem of optical coherence theory.
\newblock {\em Journal of the Optical Society of America}, 63(2):126--134,
  1973.

\bibitem{Lauer13b}
F.~Lauer and H.~Ohlsson.
\newblock Finding sparse solutions of systems of polynomial equations via
  group-sparsity optimization.
\newblock {\em arXiv preprint}, arXiv:1311.5871, 2013.

\bibitem{Mukherjee12}
S.~Mukherjee and C.~S. Seelamantula.
\newblock An iterative algorithm for phase retrieval with sparsity constraints:
  application to frequency domain optical coherence tomography.
\newblock In {\em IEEE International Conference on Acoustics, Speech and Signal
  Processing (ICASSP)}, pages 553--556, 2012.

\bibitem{Ohlsson13c}
H.~Ohlsson and Y.~C. Eldar.
\newblock On conditions for uniqueness in sparse phase retrieval.
\newblock {\em CoRR}, abs/1308.5447, 2013.

\bibitem{Ohlsson13nlbp}
H.~Ohlsson, A.~Y. Yang, R.~Dong, and S.~Sastry.
\newblock Nonlinear basis pursuit.
\newblock {\em arXiv preprint arXiv:1304.5802}, 2013.

\bibitem{Ohlsson13quadratic}
H.~Ohlsson, A.~Y. Yang, R.~Dong, M.~Verhaegen, and S.~Sastry.
\newblock Quadratic basis pursuit.
\newblock {\em arXiv preprint arXiv:1301.7002}, 2013.

\bibitem{Ranieri13}
J.~Ranieri, A.~Chebira, Y.~M. Lu, and M.~Vetterli.
\newblock Phase retrieval for sparse signals: Uniqueness conditions.
\newblock {\em CoRR}, abs/1308.3058, 2013.

\bibitem{Shechtman13}
Y.~Shechtman, A.~Beck, and Y.~C. Eldar.
\newblock {GESPAR}: Efficient phase retrieval of sparse signals.
\newblock {\em arXiv preprint arXiv:1301.1018}, 2013.

\bibitem{Shechtman11}
Y.~Shechtman, Y.~C. Eldar, A.~Szameit, and M.~Segev.
\newblock Sparsity based sub-wavelength imaging with partially incoherent light
  via quadratic compressed sensing.
\newblock {\em Optics Express}, 19(16):14807--14822, 2011.

\bibitem{Waldspurger12}
I.~Waldspurger, A.~{d'Aspremont}, and S.~Mallat.
\newblock Phase recovery, maxcut and complex semidefinite programming.
\newblock {\em arXiv preprint arXiv:1206.0102}, 2012.

\bibitem{Walther63}
A.~Walther.
\newblock The question of phase retrieval in optics.
\newblock {\em Journal of Modern Optics}, 10(1):41--49, 1963.

\end{thebibliography}


\appendix 

\section{Lemmas}
\label{sec:lemmas} 

\begin{lemma}\label{lem:nudiffreal}
Let $\nu$ and $\nu^{-1}$ be defined as in Section~\ref{sec:real}. For all $\g v\in\R^M$ such that $\nu^{-1}(\g v) \neq \g 0$ and $\g v^\prime\in\R^M$, if $\left(\nu(\nu^{-1}(\g v))\right)_{jj} \neq v_{jj}^\prime$ for some $j\in\{1,\dots,n\}$, then $\g v \neq \g v^\prime$. 
\end{lemma}
\begin{proof}
Assume that $\nu^{-1}(\g v) \neq \g 0$ and let $i =\min_j j,\ \mbox{s.t. } v_{jj} > 0$.  Then, $\forall j\in\{1,\dots,n\}$, we have
$$
	\left(\nu(\nu^{-1}(\g v))\right)_{jj} = \left(\nu^{-1}(\g v)\right)_j^2 = \frac{v_{ji}^2}{v_{ii}} = v_{jj} ,
$$
where the last equality follows from the definition of $\nu^{-1}$. Therefore, if $\left(\nu(\nu^{-1}(\g v))\right)_{jj} = v_{jj} \neq v_{jj}^\prime$ for some $j\in\{1,\dots,n\}$, we obtain $\g v \neq \g v^\prime$.
\end{proof}

\begin{lemma}\label{lem:bounddelta2}
	Let $\g A = [\g A_1,\dots, \g A_M]$ be an $N\times M$ real matrix with mutual coherence $\mu(\g A)$ as in Definition~\ref{def:mu}. Let $\g W$ be the $M\times M$-diagonal matrix of entries $w_i = \|\g A_i\|_2$. Then, for all $\g \delta \in Ker(\g A)$ and $i\in\{1,\dots,M\}$, the bound
\begin{equation}\label{eq:bounddelta2}
	w_i^2 \delta_i^2 \leq \frac{\mu^2(\g A)}{1+\mu^2(\g A)} \|\g W \g \delta\|_2^2 
\end{equation}
holds. 
\end{lemma}
\begin{proof}
See Lemma~2 in \cite{Lauer13b}.
\end{proof}

\begin{lemma}\label{lem:complexquadratic}
Let the complex Veronese map $\nu$ be as in Definition~\ref{def:veronesecomplex}. Then, for all $\g q\in\C^n$ and $\g x\in\C^n$, the following equality holds: 
$$
	\g x^H \g q \g q^H \g x = 2 \Re\left(\nu(\g q)^H\nu(\g x)\right) - \sum_{j=1}^n  \left(\nu(\g q)\right)_{jj}  \left(\nu(\g x)\right)_{jj} .
$$
\end{lemma}
\begin{proof}
\begin{align*}
	\g x^H \g q \g q^H \g x &= \left(\sum_{k=1}^n  \overline{x_k}q_k\right)\left( \sum_{j=1}^n \overline{q_j} x_j \right) \\
		& =  \sum_{j=1}^n \overline{q_j} x_j\sum_{k=1}^n q_k \overline{x_k} =  \sum_{j=1}^n \sum_{k=1}^n \overline{q_j}  q_k x_j \overline{x_k} \\
		&=  \sum_{j=1}^n q_j \overline{q_j} x_j \overline{x_j} + \sum_{j=1}^n \sum_{k\neq j} \overline{q_j}  q_k x_j \overline{x_k} \\
		&=  \sum_{j=1}^n q_j \overline{q_j} x_j \overline{x_j} + \sum_{j=1}^n \left(\sum_{k< j} \overline{q_j}  q_k x_j \overline{x_k} + \sum_{k> j} \overline{q_j}  q_k x_j \overline{x_k} \right) \\
		&=  \sum_{j=1}^n q_j \overline{q_j} x_j \overline{x_j} + \sum_{j=1}^n \sum_{k> j} \left(\overline{q_k}  q_j x_k \overline{x_j} +  \overline{q_j}  q_k x_j \overline{x_k} \right) \\
		&=  \sum_{j=1}^n q_j \overline{q_j} x_j \overline{x_j} + \sum_{j=1}^n \sum_{k> j} \left(\overline{\overline{q_j}  q_k x_j \overline{x_k} } +  \overline{q_j}  q_k x_j \overline{x_k} \right) 
\end{align*}
At this point, we use the fact that $z + \overline{z} = 2\Re(z)$, which yields
\begin{align*}
		\g x^H \g q \g q^H \g x &=  \sum_{j=1}^n q_j \overline{q_j} x_j \overline{x_j}  + 2 \sum_{j=1}^n \sum_{k> j} \Re\left( \overline{q_j}  q_k x_j \overline{x_k} \right) \\
		&=  \sum_{j=1}^n  \left(\nu(\g q)\right)_{jj}  \left(\nu(\g x)\right)_{jj}  + 2 \sum_{j=1}^n \sum_{k > j} \Re\left( \left(\overline{\nu(\g q)}\right)_{jk}  \left(\nu(\g x)\right)_{jk} \right) 
\end{align*}
Since $ \left(\nu(\g q)\right)_{jj}  \left(\nu(\g x)\right)_{jj}=q_j \overline{q_j} x_j \overline{x_j} = |q_j x_j|^2 $ is a real number, it can be introduced in the second sum as
\begin{align*}
		\g x^H \g q \g q^H \g x &=  - \sum_{j=1}^n  \left(\nu(\g q)\right)_{jj}  \left(\nu(\g x)\right)_{jj} + 2 \sum_{j=1}^n \sum_{k \geq j} \Re\left( \left(\overline{\nu(\g q)}\right)_{jk}  \left(\nu(\g x)\right)_{jk} \right)  \\
		&= - \sum_{j=1}^n  \left(\nu(\g q)\right)_{jj}  \left(\nu(\g x)\right)_{jj} + 2 \sum_{j=1}^n \sum_{k \geq j}\left[ \Re\left(\overline{\nu(\g q)}\right)_{jk}  \Re\left(\nu(\g x)\right)_{jk} - \Im\left(\overline{\nu(\g q)}\right)_{jk}  \Im\left(\nu(\g x)\right)_{jk}  \right]  \\
		&= - \sum_{j=1}^n  \left(\nu(\g q)\right)_{jj}  \left(\nu(\g x)\right)_{jj} + 2\left[\Re\left(\overline{\nu(\g q)}\right)^T  \Re\left(\nu(\g x)\right) - \Im\left(\overline{\nu(\g q)}\right)^T  \Im\left(\nu(\g x)\right)	\right]	\\
		&= - \sum_{j=1}^n  \left(\nu(\g q)\right)_{jj}  \left(\nu(\g x)\right)_{jj} + 2\left[\Re\left(\nu(\g q)\right)^T  \Re\left(\nu(\g x)\right) + \Im\left(\nu(\g q)\right)^T  \Im\left(\nu(\g x)\right)	\right]	\\
		&= 2\Re\left(\nu(\g q)^H\nu(\g x)\right) - \sum_{j=1}^n  \left(\nu(\g q)\right)_{jj}  \left(\nu(\g x)\right)_{jj} ,
\end{align*}
where the last equality is due to $\Re(\g a^H \g b) = \Re(\g a)^T \Re(\g b) + \Im(\g a)^T \Im(\g b)$. 
\end{proof}

\begin{lemma}\label{lem:sparsitynuinv}
Let $\nu^{-1}$ be as in Definition~\ref{def:nuinv} and the matrices $\g W_j$ as in Sect.~\ref{sec:firstrelax}. Then, for all $\g v\in\C^M$, 
\begin{equation}\label{eq:sparsitynuinv}
	\|\nu^{-1}(\g v)\|_0\leq \|\{\g W_j\g v\}_{j=1}^n\|_0 .
\end{equation}
\end{lemma}
\begin{proof}
According to Definition~\ref{def:nuinv}, $|\left(\nu^{-1}(\g v)\right)_j|^2 = \frac{|v_{ji}|^2}{ v_{ii} } = v_{jj}$. On the other hand, $v_{jj}$ belongs to a single group of variables generated by the $\g W_j$, i.e., $\left(\g W_k\right)_{(jj)} \neq \g 0 \Leftrightarrow k= j$. Thus, $\left(\nu^{-1}(\g v)\right)_j \neq 0 \Rightarrow v_{jj} > 0 \Rightarrow \g W_j\g v \neq \g 0 \Rightarrow \|\nu^{-1}(\g v)\|_0\leq \|\{\g W_j\g v\}_{j=1}^n\|_0$. 
\end{proof}
However, note that the converse is not true: we can have $\|\nu^{-1}(\g v)\|_0 \neq \|\{\g W_j\g v\}_{j=1}^n\|_0$, e.g., when $v_{jj} = 0$, $j=1,\dots,n$, and $v_{ij} \neq 0$ for some $i$ and $j$.

\begin{lemma}\label{lem:nudiff}
Let $\nu$ and $\nu^{-1}$ be defined as in Definitions~\ref{def:veronesecomplex} and~\ref{def:nuinv}. For all $\g v\in\C^M$ such that $\nu^{-1}(\g v) \neq \g 0$ and $\g v^\prime\in\C^M$, if $\left(\nu(\nu^{-1}(\g v))\right)_{jj} \neq v_{jj}^\prime$ for some $j\in\{1,\dots,n\}$, then $\g v \neq \g v^\prime$. 
\end{lemma}
\begin{proof}
Assume that $\nu^{-1}(\g v) \neq \g 0$ and let $i =\min_j j,\ \mbox{s.t. } \Re(v_{jj} ) > 0,\ \Im(v_{jj}) = 0 $.  Then, $\forall j\in\{1,\dots,n\}$, we have
$$
	\left(\nu(\nu^{-1}(\g v))\right)_{jj} = \left(\nu^{-1}(\g v)\right)_j \left(\overline{\nu^{-1}(\g v)}\right)_j = \frac{\overline{v_{ij}}}{\sqrt{v_{ii}}} \frac{{v_{ij}}}{\sqrt{v_{ii}}} = \frac{|v_{ij}|^2}{v_{ii}} = v_{jj} ,
$$
where the last equality follows from Definition~\ref{def:nuinv}. Therefore, if $\left(\nu(\nu^{-1}(\g v))\right)_{jj} = v_{jj} \neq v_{jj}^\prime$ for some $j\in\{1,\dots,n\}$, we obtain $\g v \neq \g v^\prime$.
\end{proof}

\begin{lemma}\label{lem:deltacomplex}
For all $\g \delta\in\C^M$ such that $\Re(\g A\g \delta) = \g 0$, the inequality
$$
	(w_{i}^R)^2 \Re(\delta_{i})^2 + (w_i^I)^2 \Im(\delta_i)^2 \leq \frac{2 \mu^2(\tilde{\g A})}{1+\mu^2(\tilde{\g A})} \|\g W^R \Re(\g \delta) + \i \g W^I \Im(\g \delta) \|^2  
$$
holds with $\tilde{\g A} = [\Re(\g A),\ -\Im(\g A)] \in \R^{N\times 2M}$. 
\end{lemma}
\begin{proof}
We rewrite the assumption as
\begin{align*}
	\g 0 &= \Re(\g A\g \delta) =  \tilde{\g A} \tilde{\g \delta} ,
\end{align*}
where $\tilde{\g \delta} =  [\Re(\g \delta)^T,\ \Im(\g \delta)^T]^T \in \R^{2M}$. 
Then, we can use Lemma~\ref{lem:bounddelta2} to bound the entries in $\tilde{\g \delta}$ as
$$
	(w_{i}^R)^2 \Re( \delta_{i})^2 \leq \frac{\mu^2(\tilde{\g A})}{1+\mu^2(\tilde{\g A})}  \left\|\tilde{ \g W} \tilde{\g \delta}\right\|_2^2 ,
$$
with $\tilde{ \g W} =  \begin{pmatrix}\g W^R & \g 0\\\g 0 & \g W^I\end{pmatrix}$ a diagonal matrix of entries $(\tilde{ \g W})_{i,i}= \|\tilde{\g A}_i\|_2$, and
$$
	(w_{i}^I)^2 \Im( \delta_{i})^2 \leq \frac{\mu^2(\tilde{\g A})}{1+\mu^2(\tilde{\g A})}   \left\|\tilde{ \g W} \tilde{\g \delta}\right\|_2^2 .
$$
These inequalities lead to
$$
	(w_{i}^R)^2 \Re(\delta_{i})^2 + (w_i^I)^2 \Im(\delta_i)^2  \leq \frac{2 \mu^2(\tilde{\g A})}{1+\mu^2(\tilde{\g A})}  \|\tilde{\g W} \tilde{ \g \delta} \|_2^2 ,
$$
where
$$
	\|\tilde{\g W} \tilde{ \g \delta} \|_2^2 = \|\g W^R \Re(\g \delta) \|_2^2 + \|\g W^I \Im(\g \delta) \|_2^2 = \|\g W^R \Re(\g \delta) + \i \g W^I \Im(\g \delta) \|^2 .
$$
\end{proof}

\begin{lemma}\label{lem:deltacomplexreal}
For all $\g \delta\in\R^M$ such that $\Re(\g A)\g \delta = \g 0$, the inequality
$$
	(w_{i}^R)^2 \delta_{i}^2 \leq \frac{ \mu^2\left(\Re(\g A)\right)}{1+\mu^2\left(\Re(\g A)\right)} \|\g W^R \g \delta \|_2^2  ,
$$
where $\g W^R$ is a diagonal matrix of entries $(\g W^R)_{i,i} = w_i^R = \|\Re(\g A_i)\|_2$. 
\end{lemma}
\begin{proof}
This is a direct consequence of Lemma~\ref{lem:bounddelta2} applied to $\Re(\g A)$ and $\g \delta\in Ker(\Re(\g A))$. 
\end{proof}

\section{Proofs}

\subsection{Proof of Theorem~\ref{thm:groupsparse}}
\label{proof:groupsparse}

This proof is similar to the one of Theorem~3 in~\cite{Lauer13b}. 
\begin{proof}
The vector $\g v_0$ is the unique solution to \eqref{eq:P1group} if  the inequality
$$
        \sum_{j=1}^n \|\g W_j\g W (\g v_0 + \g \delta)\|_2 > \sum_{j=1}^n \|\g W_j\g W \g v_0\|_2
$$
holds for all $\g \delta \neq \g 0$ satisfying $\g A\g \delta = 0$. 
The inequality above can be rewritten as
$$
        \sum_{j\in I_0} \|\g W_j\g W \g \delta\|_2 +\sum_{j\notin I_0} \|\g W_j\g W (\g v_0 + \g \delta)\|_2 -  \|\g W_j\g W \g v_0\|_2 > 0 ,
$$
where $I_0=\{j\in\{1,\dots,n\} : \g W_j\g W \g v_0 = \g 0\}$.  
By the triangle inequality, $\|\g a + \g b\|_2 - \|\g a\|_2 \geq -\|\g b\|_2$, this condition is met if
$$
        \sum_{j\in I_0} \|\g W_j\g W \g \delta\|_2 -\sum_{j\notin I_0} \|\g W_j\g W\g \delta\|_2  > 0
$$
or
\begin{equation}\label{eq:proof1}
        \sum_{j=1}^n \|\g W_j\g W\g \delta\|_2 - 2\sum_{j\notin I_0} \|\g W_j\g W\g \delta\|_2  > 0 .
\end{equation}

By defining $G_j$ as the set of indexes corresponding to nonzero columns of $\g W_j$, Lemma~\ref{lem:bounddelta2} yields
\begin{align*}
        \|\g W_j\g W\g \delta\|_2^2 &= \sum_{i\in G_j} w_i^2 \delta_i^2  \leq n\frac{\mu^2(\g A)}{1+\mu^2(\g A)} \|\g W \g \delta\|_2^2 ,
\end{align*}
Due to the fact that $\bigcup_{k\in\{1,\dots,n\}} G_k = \{1,\dots,M\}$, we also have
$$
        \|\g W \g \delta\|_2^2 = \sum_{i=1}^M w_i^2\delta_i^2 \leq  \sum_{k=1}^n \sum_{i\in G_k} w_i^2 \delta_i^2 
        = \sum_{k=1}^n \| \g W_k \g W\g \delta\|_2^2 \leq \left(\sum_{k=1}^n \| \g W_k \g W\g \delta\|_2 \right)^2 ,
$$
which then leads to
\begin{align*}
        \|\g W_j\g W\g \delta\|_2^2 
        & \leq n\frac{\mu^2(\g A)}{1+\mu^2(\g A)}  \left( \sum_{k=1}^n \| \g W_k \g W\g \delta\|_2\right)^2 .
\end{align*}
Introducing this result in \eqref{eq:proof1} gives the condition
$$
        \sum_{j=1}^n \|\g W_j\g W\g \delta\|_2 - 2 (n - |I_0|) \frac{\mu(\g A) \sqrt{n}}{\sqrt{1+\mu^2(\g A)}} \sum_{k=1}^n \| \g W_k\g W \g \delta\|_2   > 0 .
$$
Finally, given that $|I_0| = n - \|\{\g W_j\g v_0\}_{j=1}^n\|_0 = n - \|\g x_0\|_0$, this yields
$$
        \sum_{j=1}^n \|\g W_j\g W\g \delta\|_2 - 2 \|\g x_0\|_0 \frac{\mu(\g A) \sqrt{n}}{\sqrt{1+\mu^2(\g A)}} \sum_{k=1}^n \| \g W_k\g W \g \delta\|_2   > 0 .
$$
or, after rearranging the terms,  
$$
         \|\g x_0\|_0 < \frac{\sqrt{1+\mu^2(\g A)}}{2\mu(\g A) \sqrt{n} } ,
$$
which can be rewritten as in the statement of the Theorem. 
\end{proof}

\subsection{Proof of Theorem~\ref{thm:groupsparsecomplexreal}}
\label{sec:proofthmcomplexreal}

This proof is very similar to the ones of Theorems~\ref{thm:groupsparse} and~\ref{thm:groupsparsecomplex}.
\begin{proof}
The vector $\g v_0$ is the unique solution to \eqref{eq:P1vcomplexreal} if  the inequality
$$
	\sum_{j=1}^n \|\g W_j^R (\g v_0 + \g \delta) \|_2 > \sum_{j=1}^n \|\g W_j^R \g v_0\|_2
$$
holds for all $\g \delta\in\R^M$ such that $\Re(\g A)(\g v_0 + \g \delta) = \g y$, which implies the constraint $ \Re\left(\g A\right)\g \delta = \g 0$ on $\g\delta$.
The inequality above can be rewritten as
$$
	\sum_{j\in I_0} \|\g W_j^R \g \delta\|_2 + \sum_{j\notin I_0} \|\g W_j^R (\g v_0 + \g \delta) \|_2 -  \|\g W_j^R \g v_0 \|_2 > 0 ,
$$
where $I_0=\{j\in\{1,\dots,n\} : \g W_j^R \g v_0 = \g 0 \}$.  
By the triangle inequality, $\|\g a + \g b\| - \|\g a\| \geq -\|\g b\|$ with $\g a = \g W_j^R \g v_0$, this condition is met if
\begin{equation}\label{eq:proof1cplxreal}
	\sum_{j=1}^n \|\g W_j^R \g \delta \|_2  - 2\sum_{j\notin I_0} \|\g W_j^R \g \delta\|_2 > 0 .
\end{equation}

By defining $G_j$ as the set of indexes corresponding to nonzero columns of $\g W_j$, Lemma~\ref{lem:deltacomplexreal} yields
\begin{align*}
	\|\g W_j^R \g \delta \|_2^2 &= \sum_{i\in G_j} (w_i^R)^2 \delta_i^2  
	 \leq n\frac{\mu^2(\Re(\g A))}{1+\mu^2(\Re(\g A))} \|\g W^R \g \delta \|_2^2  .
\end{align*}
Due to the fact that $\bigcup_{k\in\{1,\dots,n\}} G_k = \{1,\dots,M\}$, we also have
\begin{align*}
	\|\g W^R \g \delta\|_2^2  &= \sum_{i=1}^M (w_i^R)^2 \delta_i^2  
	\leq  \sum_{k=1}^n \sum_{i\in G_k} (w_i^R)^2 \delta_i^2  
	= \sum_{k=1}^n \|\g W_k^R  \g \delta\|_2^2 
	 \leq \left(\sum_{k=1}^n \|\g W_k^R \g \delta\|_2 \right)^2 ,
\end{align*}
which then leads to
\begin{align*}
	\|\g W_j^R \g \delta\|_2^2
	& \leq n\frac{\mu^2(\Re(\g A))}{1+\mu^2(\Re(\g A))} \left(\sum_{k=1}^n \|\g W_k^R \g \delta\|_2 \right)^2.
\end{align*}
Introducing this result in \eqref{eq:proof1cplxreal} gives the condition
$$
	\sum_{j=1}^n \|\g W_j^R  \g \delta\|_2  - 2 (n - |I_0|) \frac{\mu(\Re(\g A)) \sqrt{n}}{\sqrt{1+\mu^2(\Re(\g A))}} \sum_{k=1}^n \|\g W_k^R  \g \delta\|_2  > 0 .
$$
Finally, given that $|I_0| = n - \|\g x_0\|_0$, this yields
, for $\g\delta\neq \g 0$,
$$
	 \|\g x_0\|_0 < \frac{\sqrt{1+\mu^2(\Re(\g A))}}{2\mu(\Re(\g A)) \sqrt{n} } ,
$$
which can be rewritten as in the statement of the Theorem. 
\end{proof}

\subsection{Proof of Proposition~\ref{prop:varphi}}
\label{sec:proofprop}

\begin{proof}
To prove Statement 1, note that the operation $\g x_2 = \mbox{shift}(\mbox{reflection}(\g x_1), 1)$ is equivalent to 
$$
	\g x_2 = \begin{bmatrix} x_{11} \\ \mbox{reflection}(\tilde{\g x}_1)\end{bmatrix}, \quad \tilde{\g x}_1 = [ x_{12}, x_{13}, \dots, x_{1n}]^T,
$$
where the reflection is centered on $\tilde{x}_{1(n/2)} = x_{1(1+ n/2)}$. 
Thus, 
$$
	\g x_2 = [x_{11} , x_{1n}, x_{1(n-1)}, \dots, x_{12}]^T
$$
and 
$$
	\sum_{i=2}^{n/2} |x_{2i}|^2 =  \sum_{i=0}^{n/2 - 2} |x_{1(n-i)}|^2 =  \sum_{i=2+\frac{n}{2}}^{n} |x_{1i}|^2 
	>
	\sum_{i=2}^{n/2} |x_{1i}|^2 =  \sum_{i=0}^{n/2 - 2} |x_{2(n-i)}|^2 = \sum_{i=2+\frac{n}{2}}^{n} |x_{2i}|^2 ,
$$
where the inequality holds whenever this operation is performed in~\eqref{eq:phireflect} (i.e., if this is not the case, then $\g x_2 = \g x_1$). 
Therefore, for all $\g x$, we obtain a vector $\g x_2=\varphi(\g x)$ such that $\sum_{i=2}^{n/2} |x_{2i}|^2 \geq \sum_{i=2+\frac{n}{2}}^{n} |x_{2i}|^2$. Since $\arg\max_{i\in\{1,\dots,n\}} |x_{2i}| = 1$ and $\mbox{shift}(\g x_2, 1 - \arg\max_{i\in\{1,\dots,n\}} |x_{2i}| ) = \g x_2$, this implies that $\varphi(\g x_2) = \g x_2$, i.e., $\varphi(\varphi(\g x)) = \varphi(\g x)$.

Statement 2 is easily seen from the fact that $\g x_1$ does not change when $\g x$ is shifted. 

To prove Statement 3, let $\g x^\prime = \mbox{reflection}(\g x) = [x_n, x_{n-1},\dots, x_1]^T$, $k = \arg\max_{i\in\{1,\dots, n\}} |x_i|$, and $k^\prime = \arg\max_{i\in\{1,\dots, n\}} |x_i^\prime|$. Then, $k^\prime = n - k + 1$ and 
\begin{align*}
	\g x_1^\prime &= \mbox{shift}(\g x^\prime, 1-k^\prime) = \mbox{shift}(\g x^\prime, k-n) = \mbox{shift}(\g x^\prime, k)  \\
	&= [x_{n-k+1}^\prime, x_{n-k+2}^\prime, \dots, x_{n}^\prime, x_{1}^\prime, x_2^\prime, \dots, x_{n-k}^\prime]^T \\
	&=  [x_k, x_{k-1}, \dots, x_1, x_n, x_{n-1}, \dots, x_{k+1}]^T.
\end{align*}

Define $\g x_1 = \mbox{shift}(\g x, 1-k) = [x_k, x_{k+1}, \dots, x_n, x_1, x_{2}, \dots, x_{k-1}]^T$. 

If $k< n/2$, we have
$$
	\sum_{i=2}^{n/2} |x_{1i}^\prime|^2 = \sum_{i=1}^{k-1} |x_i|^2 + \sum_{i=1}^{\frac{n}{2} -k } |x_{n-i+1}|^2 
	= \sum_{i=n+2-k}^{n} |x_{1i}|^2  +   \sum_{i=n/2 + 2}^{n-k+1} |x_{1i}|^2
	=\sum_{i=n/2 + 2}^{n} |x_{1i}|^2,
$$ 
and otherwise, if $k\geq n/2$, we obtain
$$
	\sum_{i=2}^{n/2} |x_{1i}^\prime|^2 = \sum_{i=k - \frac{n}{2} + 1}^{k -1} |x_i|^2  
	= \sum_{i=n/2+2}^{n} |x_{1i}|^2  .
$$ 

On the other hand, if $k\leq n/2$, then
$$
	\sum_{i=2+\frac{n}{2}}^{n} |x_{1i}^\prime|^2 = \sum_{i=k+1}^{\frac{n}{2} + k - 1} |x_i|^2 
	=  \sum_{i=2}^{n/2} |x_{1i}|^2  ,
$$
while if $k> n/2$:
$$
	\sum_{i=2+\frac{n}{2}}^{n} |x_{1i}^\prime|^2 = \sum_{i=k+1}^{n} |x_i|^2  +  \sum_{i=1}^{k - \frac{n}{2} - 1} |x_i|^2 
	= \sum_{i=2}^{n-k+1} |x_{1i}|^2   +   \sum_{i=n-k+2}^{n/2} |x_{1i}|^2 =  \sum_{i=2}^{n/2} |x_{1i}|^2  .
$$

Therefore, 
$$
\sum_{i=2}^{n/2} |x_{1i}|^2 > \sum_{i=2+\frac{n}{2}}^{n} |x_{1i}|^2 
\quad \Leftrightarrow\quad \sum_{i=2}^{n/2} |x_{1i}^\prime|^2 < \sum_{i=2+\frac{n}{2}}^{n} |x_{1i}^\prime|^2
$$
and $\varphi$ applies a reflection to $\g x$ if and only if it does not apply one to $\g x^\prime$ (the inequalities above are strict by assumption). 
Thus, in the case $\g x$ is reflected by $\varphi$, we have
$$
	\varphi(\g x^\prime) = \mbox{shift}(\mbox{reflection}(\g x),k) = \mbox{reflection}( \mbox{shift}(\g x, n-k ) )
$$
and
\begin{align*}
	\varphi(\g x) &= \mbox{shift}(\mbox{reflection}( \mbox{shift}(\g x,1-k) ), 1) \\
		&= \mbox{reflection}( \mbox{shift}(\mbox{shift}(\g x,1-k), n -1 ) ) \\
			&=  \mbox{reflection}( \mbox{shift}(\g x,n-k) )\\
			&= \varphi(\g x^\prime) ,
\end{align*}
while in the case $\g x$ is not reflected by $\varphi$, we have
$$
	\varphi(\g x) = \mbox{shift}(\g x,1-k)
$$
and
\begin{align*}
	\varphi(\g x^\prime) &= \mbox{shift}(\mbox{reflection}( \mbox{shift}(\mbox{reflection}(\g x),k) ), 1) \\
		&= \mbox{reflection}( \mbox{shift}(\mbox{shift}(\mbox{reflection}(\g x),k) , n -1 ) ) \\
			&=  \mbox{reflection}( \mbox{shift}(\mbox{reflection}(\g x),n+k-1) )\\
			&=  \mbox{reflection}( \mbox{reflection}(\mbox{shift}(\g x, - k + 1) ) )\\
			&= \mbox{shift}(\g x, 1-k)\\
			&= \varphi(\g x) .
\end{align*}

\end{proof}

\end{document}